\documentclass[onecolumn,11pt,draftcls]{IEEEtran}%
\usepackage{amsbsy}
\usepackage{floatflt} 

\usepackage{amsmath}
\usepackage{amssymb}
\usepackage{times}
\usepackage{graphicx}
\usepackage{xspace}
\usepackage{paralist} 
\usepackage{setspace} 
\usepackage{xypic}
\xyoption{curve}
\usepackage{latexsym}
\usepackage{theorem}
\usepackage{ifthen}
\usepackage{subfigure}
\usepackage[ruled,vlined]{algorithm2e}
\usepackage{booktabs}
\usepackage{algorithmic}

\newcommand{\mypar}[1]{\vspace{0.03in}\noindent{\bf #1.}}

{\theoremheaderfont{\it} \theorembodyfont{\rmfamily}
\newtheorem{theorem}{Theorem}
\newtheorem{lemma}[theorem]{Lemma}

}

\title{Weight Optimization for Consensus Algorithms with Correlated Switching Topology}
\author{Du$\breve{\mbox{s}}$an Jakoveti\'c, Jo\~ao Xavier, and Jos\'e M.~F.~Moura$^{*}$
\thanks{The first and second authors are with the Instituto de Sistemas e Rob\'otica~(ISR), Instituto Superior T\'ecnico~(IST), 1049-001 Lisboa, Portugal. The first and third authors are with the Department of Electrical and Computer Engineering,
Carnegie Mellon University, Pittsburgh, PA 15213, USA (e-mail:
[djakovetic,jxavier]@isr.ist.utl.pt, moura@ece.cmu.edu, ph: (412)268-6341, fax: (412)268-3890.)}
\thanks{Work partially supported by  NSF under grant~\#~CNS-0428404, by  the Office of
Naval Research under MURI N000140710747, and by the Carnegie Mellon$|$Portugal Program under a grant of the Funda\,c\~ao de Ci$\hat{\mbox{e}}$ncia e Tecnologia~(FCT) from Portugal. Du$\breve{\mbox{s}}$an Jakovetic holds a fellowship from~FCT.}}
\begin{document}
\maketitle \thispagestyle{empty} \maketitle
\begin{abstract}
We design the weights in consensus algorithms with spatially correlated random topologies. These arise with:  1) networks with spatially correlated random link failures and 2) networks with randomized averaging protocols. We show that the weight optimization problem is convex for both symmetric and asymmetric random graphs.
With symmetric random networks, we choose the consensus mean squared error (MSE) convergence rate as optimization criterion and explicitly
express this rate as a function of the link formation probabilities, the link formation spatial correlations, and the consensus weights. We prove that the MSE convergence rate is a convex, nonsmooth function of the weights, enabling global optimization of the weights for arbitrary link formation probabilities and link correlation structures. We extend our results to the case of asymmetric random links. We adopt as optimization criterion the mean squared deviation (MSdev) of the nodes' states from the current average state. We prove that MSdev is a convex function of the weights.
Simulations show that significant performance gain is achieved with our weight design method when compared with methods available in the literature.
\end{abstract}

\hspace{.43cm}\textbf{Keywords:} Consensus, weight optimization, correlated link failures, unconstrained optimization, sensor networks, switching topology, broadcast gossip.
\newpage

\section{Introduction}
\label{Intro}

This paper finds the optimal weights for the consensus algorithm in correlated random networks. Consensus is an iterative distributed algorithm that computes the global average of data distributed among a network of agents using only local communications. Consensus has renewed interest in distributed algorithms (\cite{Tsitsiklis,tsitsiklisThesis84}), arising in many different areas from distributed data fusion (\cite{SoummyaRamanujan,MouraInference,BoydFusion,giannakis-tsp,SoummyaEst}) to coordination of mobile autonomous agents (\cite{jabaNeghbor,Murray}).
 A recent survey is~\cite{MurraySwitching}.

This paper studies consensus algorithms in networks where the links (being online or off line)
are random. We consider two scenarios: 1) the network is random, because links in the network may fail at random times; 2) the network protocol is randomized, i.e., the link states along time are controlled by a randomized protocol (e.g., standard gossip algorithm~\cite{BoydGossip}, broadcast gossip algorithm~\cite{scaglione}). In both cases, we model the links as Bernoulli random variables. Each link has some formation probability, i.e., probability of being active, equal to $P_{ij}$. Different links may be correlated at the same time, which can be expected in real applications. For example, in wireless sensor networks (WSNs) links can be spatially correlated due to interference among close links or electromagnetic shadows that may affect several nearby sensors.

References on consensus under time varying  or random topology are (\cite{TsitsiklisFlocking,MurraySwitching,BoydMetropolis}) and (\cite{JadbabaieErgodic,SoummyaConferenceConsensus,Mesbahi,Kontinualni, scaglione}), among others, respectively. Most of the previous work is focussed on providing convergence conditions and/or characterizing the convergence rate under different assumptions on the network randomness (\cite{Mesbahi,SoummyaConferenceConsensus,Kontinualni}). For example, references~\cite{SoummyaConferenceConsensus} and~\cite{SoummyaChannelNoise} study consensus algorithm with spatially and temporally independent link failures. They show that a necessary and sufficient condition for mean squared and almost sure convergence is for the communication graph to be connected on average.

We consider here the weight optimization problem: how to assign the weights $W_{ij}$ with which the nodes mix their states across the network, so that the convergence towards consensus is the fastest possible. This problem has not been solved (with full generality) for consensus in random topologies.
We study this problem for networks with symmetric and asymmetric random links separately, since the properties of the corresponding algorithm are different. For symmetric links (and connected network topology on average), the consensus algorithm converges to the average of the initial nodes' states almost surely. For asymmetric random links, all the nodes asymptotically reach agreement, but they only agree to a random variable in the neighborhood of the true initial average.

We refer to our weight solution as probability-based weights (PBW).
PBW are simple and suitable for distributed implementation: we assume at each iteration that the weight of link $(i,j)$ is~$W_{ij}$ (to be optimized), when the link is alive, or 0, otherwise. Self-weights are adapted such that the row-sums of the weight matrix at each iteration are one. This is suitable for distributed implementation. Each node  updates readily after receiving messages from its current neighbors. No information about the number of nodes in the network or the neighbor's current degrees is needed. Hence, no additional online communication is required for computing weights, in contrast, for instance, to the case of the Metropolis weights (MW)~\cite{BoydMetropolis}.

Our weight design method assumes that the link formation probabilities and their spatial correlations are known.
With randomized protocols, the link formation probabilities and their correlations are induced by the protocol itself, and thus are known.
For networks with random link failures, the link formation probabilities relate to the signal to noise ratio at the receiver and can be computed.  In~\cite{SoummyaTopologyDesign}, the formation probabilities are designed in the presence of link communication costs and an overall network communication cost budget.
  When the WSN infrastructure is known, it is possible to estimate the link formation \emph{probabilities} by measuring the reception rate of a link computed as the ratio between the number of received and the total number of sent packets. Another possibility is to estimate the link formation probabilities based on the received signal strength. Link formation \emph{correlations} can also be estimated on actual WSNs, \cite{ZhaoPackets}. If there is no training period to characterize quantitatively the links on an actual WSN, we can still model the probabilities and the correlations as a function of the transmitted power and the inter-sensor distances. Moreover, several empirical studies (\cite{ZhaoPackets,Potkonjic} and references therein) on the quantitative properties of wireless communication in sensor networks have been done that provide models for packet delivery performance in WSNs.

\mypar{Summary of the paper} Section~\ref{Contribution_and_Related_Work} lists our contributions, relate them with
the existing literature, and introduces notation used in the paper. Section~\ref{Problem_Model} describes our model of random networks and the consensus algorithm.
 Sections~\ref{weight_optimization_symmetric} and~\ref{weight_optimization_asymmetric} study the weight optimization for symmetric random graphs and asymmetric random graphs, respectively. Section~\ref{simulations} demonstrates the effectiveness of our approach with simulations.
  Finally, section~\ref{ConclusionSection} concludes the paper. We derive the proofs of some results in the Appendices A through C.


\section{Contribution, Related Work, and Notation}
\label{Contribution_and_Related_Work}
\mypar{Contribution}
Building our results on the previous extensive studies of convergence conditions and rates for consensus algorithm (e.g.,\cite{scaglione, JadbabaieErgodic, SoummyaTopologyDesign}), we address the problem of weights optimization in consensus algorithms with correlated random topologies. Our method is applicable to:
 1) networks with correlated random link failures (see, e.g.,~\cite{SoummyaTopologyDesign} and 2) networks with randomized algorithms (see, e.g,~\cite{BoydGossip, scaglione}). We first address the weight design problem for symmetric random links, and then extend the results to asymmetric random links.

With symmetric random links, we use the mean squared consensus convergence rate $\phi(W)$ as the optimization criterion. We explicitly
express the rate $\phi(W)$ as a function of the link formation probabilities, their correlations, and the weights. We prove that $\phi(W)$ is a convex,
nonsmooth function of the weights. This enables global optimization of the weights for arbitrary link formation probabilities and and arbitrary link correlation structures.
We solve numerically the resulting optimization problem by subgradient algorithm, showing also that the optimization computational cost grows tolerably with the network size.
We provide insights into weight design with a simple example of complete random network that admits closed form solution for the optimal weights and
 convergence rate and show how the optimal weights depend on the number of nodes, the link formation probabilities, and their correlations.

We extend our results to the case of asymmetric random links, adopting as an optimization criterion the mean squared deviation (from the
current average state) rate $\psi(W)$, and show that $\psi(W)$ is a convex function of the weights.

We provide comprehensive simulation experiments to demonstrate the effectiveness of our approach. We provide
two different models of random networks with correlated link failures; in addition, we study the broadcast gossip algorithm~\cite{scaglione},
 as an example of randomized protocol with asymmetric links. In all cases, simulations confirm that our method shows significant gain compared to the methods available in the literature. Also, we show that the gain increases with the network size.

\mypar{Related work}
Weight optimization for consensus with switching topologies has not received much attention in the literature.
Reference~\cite{SoummyaTopologyDesign} studies the tradeoff between the convergence rate and the amount of communication that takes place in the network. This reference is mainly concerned with the design of the network topology,  i.e., the design of the probabilities of reliable communication $\{{P}_{ij}\}$ and the weight $\alpha$ (assuming all nonzero weights are equal), assuming a communication cost $C_{ij}$ per link and an overall network communication budget. 
Reference~\cite{scaglione} proposes the broadcast gossip algorithm, where at each time step, a single node, selected at random, broadcasts unidirectionally its state to all the neighbors within its wireless range. We detail the broadcast gossip in subsection~\ref{subsect_broadcast_gossip}. This reference optimizes the weight for the broadcast gossip algorithm assuming equal weights for all links.

The problem of optimizing the weights for consensus under a random topology, when the weights for different links may be different, has not received much attention in the literature. Authors have proposed weight choices for random or time-varying networks~\cite{WhichShouldI,BoydMetropolis}, but no claims to optimality are made.
Reference~\cite{BoydMetropolis} proposes the Metropolis weights (MW), based on the Metropolis-Hastings algorithm for
simulating a Markov chain with uniform equilibrium distribution~\cite{Hastings}.
The weights choice in~\cite{WhichShouldI} is based on the fastest mixing Markov chain problem studied in~\cite{BoydChains} and uses the information about the underlying supergraph. We refer to this weight choice as the supergraph based weights~(SGBW).

\mypar{Notation}
\label{Notation}
Vectors are denoted by a lower case letter (e.g., $x$) and it is understood from the context if $x$ denotes a deterministic or random vector.
Symbol ${\mathbb R}^{N}$ is the $N$-dimensional Euclidean space.
Inequality $x \leq y$ is understood element wise, i.e., it is equivalent to $x_i \leq y_i$, for all $i$.
Constant matrices are denoted by capital letters (e.g., $X$) and random matrices are denoted by calligraphic letters (e.g., $\mathcal{X}$). A sequence of random matrices is denoted by $\left\{ \mathcal{X}(k)  \right\}_{k=0}^{\infty}$ and the random matrix indexed by $k$ is denoted $\mathcal{X}(k)$. If the distribution of $\mathcal{X}(k)$ is the same for any $k$, we shorten the notation $\mathcal{X}(k)$ to $\mathcal{X}$ when the time instant $k$ is not of interest.
Symbol ${\mathbb R}^{N \times M}$ denotes the set of $N \times M$ real valued matrices and ${\mathbb S}^{N}$ denotes the set of symmetric real valued $N \times N$ matrices. The $i$-th column of a matrix $X$ is denoted by $X_i$. Matrix entries are denoted by $X_{ij}$. Quantities $X \otimes Y$, $X \odot Y$, and $X \oplus Y$ denote the Kronecker product, the Hadamard product, and the direct sum of the matrices $X$ and $Y$, respectively.
Inequality $X \succeq Y $ ($X \preceq Y $) means that the matrix $X-Y$ is positive (negative) semidefinite.
Inequality $X \geq Y $ ($X \leq Y $) is understood entry wise, i.e., it is equivalent to  $X_{ij} \geq Y_{ij} $, for all $i$, $j$.
Quantities $\| X \|$, $\lambda_{\mathrm{max}}( X )$, and $r(X)$ denote the matrix 2-norm, the maximal eigenvalue, and the spectral radius of $X$, respectively. The identity matrix is~$I$. Given a matrix~$A$, $\mbox{Vec}(A)$ is the column vector that stacks the columns of~$A$.
For given scalars $x_1,...,x_N$, $\mathrm{diag} \left(x_1,...,x_N  \right)$ denotes the diagonal $N\times N$ matrix with the $i$-th diagonal entry equal to $x_i$. Similarly, $\mathrm{diag}(x)$ is the diagonal matrix whose diagonal entries are the elements of~$x$. The matrix $\mathrm{diag}\left( X \right)$ is a diagonal matrix with the diagonal equal to the diagonal of~$X$.
The $N$-dimensional column vector of ones is denoted with $1$. Symbol $J=\frac{1}{N}11^T$. The $i$-th canonical unit vector, i.e., the $i$-th column of $I$, is denoted by $e_i$. Symbol $|S|$ denotes the cardinality of a set $S$.

\section{Problem model}
\label{Problem_Model}
%
%
%
This section introduces the random network model that we apply to networks with link failures and to networks
with randomized algorithms. It also introduces the consensus algorithm and the corresponding weight rule assumed in this paper.
\subsection{Random network model: symmetric and asymmetric random links}
\label{network_connectivity}
We consider random networks$\--$networks with random links or with a random protocol. Random links arise because of packet loss or drop, or
when a sensor is activated from sleep mode at a random time.
Randomized protocols like standard pairwise gossip~\cite{BoydGossip} or broadcast gossip~\cite{scaglione} activate links randomly. This section describes the network model that applies to both problems. We assume that the links are up or down (link failures) or selected to use (randomized gossip) according to spatially correlated Bernoulli random variables.

To be specific, the network is modeled by a graph $G=(V,E)$, where the set of nodes $V$ has cardinality $|V|=N$ and the set of
 directed edges $E$, with $|E|=2M$, collects all possible ordered node pairs that can communicate, i.e., all realizable links. For example, with geometric graphs,
 realizable links connect nodes within their communication radius. The graph $G$ is called supergraph, e.g.,~\cite{SoummyaTopologyDesign}. The directed edge $(i,j) \in E$ if node $j$ can transmit to node $i$.

The supergraph $G$ is assumed to be connected and without loops. For the fully connected supergraph, the number of directed edges (arrows) $2M$ is equal to
$N(N-1)$. We are interested in sparse supergraphs, i.e., the case when $M \ll \frac{1}{2}N(N-1)$.

%
%
%
%
%
%

Associated with the graph $G$ is its $N \times N$ adjacency matrix $A$:
\begin{equation*}
  {A}_{ij} = \left\{ \begin{array}{rl}
  1 &\mbox{ if $(i,j) \in E$ }\\
  0 &\mbox{otherwise}
       \end{array} \right.
\end{equation*}
The in-neighborhood set~$\Omega_{i}$ (nodes that can transmit to node $i$) and the in-degree~$d_i$ of a node~$i$ are
\begin{eqnarray*}
\Omega_{i}   &=& \left\{ j: (i,j) \in {E} \right\}\\
d_{i}   &=& |\Omega_{i}|.
\end{eqnarray*}
We model the connectivity of a random WSN at time step $k$ by a (possibly) directed random graph $\mathcal{G}(k) = \left( {V}, \mathcal{E}(k) \right)$. The random edge set is
$$\mathcal{E}(k)=\left\{ (i,j) \in E:\,\, (i,j) \,\,\mathrm{is\,\,online\,\,at\,\,time\,\,step\,\,k} \right\},$$
with $\mathcal{E}(k)\subseteq E$. The random adjacency matrix associated to $\mathcal{G}(k)$ is denoted by
$\mathcal{A}(k)$ and the random in-neighborhood for sensor $i$ by $\Omega_i(k)$.

We assume that link failures are \emph{temporally independent} and \emph{spatially correlated}.                                                                        That is, we assume that the random matrices $\mathcal{A}(k), k=0,1,2,...$ are independent identically distributed.
The state of the link $(i,j)$ at a time step $k$ is a Bernoulli random variable, with mean $P_{ij}$, i.e., $P_{ij}$ is the
 formation probability of link $(i,j)$.
At time step $k$, different edges $(i,j)$ and $(p.q)$ may be correlated, i.e., the entries $\mathcal{A}_{ij}(k)$ and $\mathcal{A}_{pq}(k)$ may be correlated. For the link $r$,  by which node $j$ transmits to node $i$, and for the link $s$, by which node $q$ transmits to node $p$, the corresponding
cross-variance is
\[
\left[ R_q \right]_{rs} = \mathrm{E} \left[  \mathcal{A}_{ij} \mathcal{A}_{pq} \right] - P_{ij} P_{pq}.
\]
Time correlation, as spatial correlation, arises naturally in many scenarios, such as when nodes awake from the sleep schedule. However,
 it requires approach different than the one we pursue in this paper~\cite{SoummyaChannelNoise}. We plan to
address the weight optimization with temporally correlated links in our future work.
\subsection{Consensus algorithm}
\label{consensus-algorithm}
%
%

Let $x_i(0)$ represent some scalar measurement or initial data available at sensor $i$, $i=1,...,N$.
Denote by $x_{\mbox{\scriptsize{avg}}}$ the average:
$$x_{\mbox{\scriptsize{avg}}} = \frac{1}{N} \sum_{i=1}^N x_i(0)$$
The consensus algorithm computes $x_{\mbox{\scriptsize{avg}}}$ iteratively at each sensor $i$ by the distributed weighted average:
\begin{equation}
x_i(k+1)= \mathcal{W}_{ii}(k)x_i(k)+ \sum_{j \in {\Omega}_i(k)} \mathcal{W}_{ij}(k) x_j(k)
\label{Eqn_consensus_ij}
\end{equation}

We assume that the random weights $\mathcal{W}_{ij}(k)$ at iteration $k$ are given by:
\begin{equation}
\mathcal{W}_{ij}(k) = \left\{ \begin{array}{rl}
 W_{ij} &\mbox{ if $j\in {\Omega}_i(k)$ } \\
  1-\sum_{m\in {\Omega}_i(k)}\mathcal{W}_{im}(k) &\mbox{ if $i=m$}\\
  0 &\mbox{ otherwise }
       \end{array} \right.
       \label{Eqn_PBW}
\end{equation}
In~\eqref{Eqn_PBW}, the quantities $W_{ij}$ are non random and will be the variables to be optimized in our work. We also take $W_{ii}=0$, for all $i$. By~\eqref{Eqn_PBW}, when the link is active, the weight is $W_{ij}$, and when not active it is zero.
Note that $W_{ij}$ are non zero only for edges $(i,j)$ in the supergraph $G$. If an edge $(i,j)$
 is not in the supergraph the corresponding $W_{ij}=0$ and $\mathcal{W}_{ij}(k) \equiv 0$.

We write the consensus algorithm in compact form.
 Let $x(k)=(x_1(k)\,\,x_2(k)\,\,...\,\,x_N(k))^T$, $W=\left[ W_{ij} \right]$, $\mathcal{W}(k) = \left[ \mathcal{W}_{ij}(k) \right]$.
The random weight matrix $\mathcal{W}(k)$ can be written in compact form as
\begin{equation}
\label{WEIGHTMATRIXCOMPACT}
\mathcal{W}(k)  =  W \odot \mathcal{A}(k)  -  \mathrm{diag} \left( W \mathcal{A}(k)  \right) + I
\end{equation}
and the consensus algorithm is simply stated with $x(k=0)=x(0)$ as
\begin{eqnarray}
\label{Eqn_consensus_Komapktno}
x(k+1) = \mathcal{W}(k) x(k),\,\,k \geq 0
\end{eqnarray}
To implement the update rule, nodes need to know their random in-neighborhood ${\Omega}_i(k)$ at every iteration.
 In practice, nodes determine ${\Omega}_i(k)$ based on who they receive messages from at iteration $k$.

It is well known~\cite{scaglione,JadbabaieErgodic} that, when the random matrix $\mathcal{W}(k)$ is symmetric, the consensus algorithm is average preserving, i.e., the sum of the states $x_i(k)$, and so the average state over time, does not
 change, even in the presence of random links. In that case the consensus algorithm converges almost surely to the true average $x_{\mbox{\scriptsize{avg}}}$. When the matrix $\mathcal{W}(k)$ is not symmetric, the average state is not preserved in time, and the state
 of each node converges to the same random variable with bounded mean squared error from $x_{\mbox{\scriptsize{avg}}}$~\cite{scaglione}.
 For certain applications, where high precision on computing the average $x_{\mbox{\scriptsize{avg}}}$ is required, average preserving, and thus a symmetric matrix $\mathcal{W}(k)$ is desirable. In practice, a symmetric matrix $\mathcal{W}(k)$ can be established by protocol design even if the underlying physical channels are asymmetric. This can be realized by ignoring unidirectional communication channels. This can be done, for instance, with a double acknowledgement protocol. In this scenario, effectively, the consensus algorithm sees the underlying random network as a symmetric network, and this scenario falls into the framework of our studies of symmetric links (section~\ref{weight_optimization_symmetric}).

 When the physical communication channels are asymmetric, and the error on the asymptotic consensus limit $c$ is tolerable, consensus with an asymmetric weight matrix $\mathcal{W}(k)$ can be used. This type of algorithm is easier to implement, since there is no need for acknowledgement protocols. An example of such a protocol is the broadcast gossip algorithm proposed in~\cite{scaglione}. Section~\ref{weight_optimization_asymmetric} studies this  type of algorithms.

\textbf{Set of possible weight choices: symmetric network.}
With symmetric random links, we will always assume $W_{ij}=W_{ji}$. By doing this we easily achieve the desirable property that $\mathcal{W}(k)$ is symmetric. The set of all possible weight choices for symmetric random links
 $S_W$ becomes:
\begin{equation}
\label{eqn_S_W_symmetric}
S_W = \left \{  W \in {\mathbb R}^{N \times N}: \, W_{ij}=W_{ji}, \,\,W_{ij}=0,\,\,\mathrm{if}\, (i,j) \notin E , \,\,W_{ii}=0, \,\,  \forall i, \right\}
\end{equation}

\textbf{Set of possible weight choices: asymmetric network.} With asymmetric random links, there is no good reason to require that $W_{ij}=W_{ji}$,
 and thus we drop the restriction $W_{ij}=W_{ji}$. The set of possible weight choices in this case becomes:
 \begin{equation}
\label{eqn_S_W_symmetric}
S_W^{\mbox{\scriptsize{asym}}} = \left \{  W \in {\mathbb R}^{N \times N}:  \,\,W_{ij}=0,\,\,\mathrm{if}\, (i,j) \notin E , \,\,W_{ii}=0, \,\,  \forall i, \right\}
\end{equation}
Depending whether the random network is symmetric or asymmetric, there will be two error quantities that will play a role. These will be discussed in detail in sections~\ref{weight_optimization_symmetric} and V, respectively. We introduce them here briefly, for reference.

\textbf{Mean square error (MSE): symmetric network.}
Define the consensus error vector $e(k)$ and the error covariance matrix $\Sigma(k)$:
\begin{eqnarray}
\label{eqn_error_vector}
e(k)&=&x(k)-x_{\mbox{\scriptsize{avg}}}1\\
\label{eqn_Sigma_k}
\Sigma(k)&=&
\mathrm{E} \left[ e(k) e(k) ^T \right].
\end{eqnarray}
The mean squared consensus error $\mathrm{MSE}$ is given by:
\begin{equation}
\label{eqn_MSE}
\mathrm{MSE}(k) = \sum_{i=1}^N \mathrm{E} \left[ \left(  x_i(k) - x_{\mbox{\scriptsize{avg}}}\right)^2 \right]  = \mathrm{E}\left[e(k)^Te(k)\right]=\mathrm{tr}\,\Sigma(k)
\end{equation}
\textbf{Mean square deviation (MSdev): asymmetric network.}
As explained, when the random links are asymmetric (i.e., when $\mathcal{W}(k)$ is not symmetric), and if the underlying supergraph is strongly connected, then the states of all nodes converge to a common value $c$ that is in general a random variable that depends on the sequence of network realizations and on the initial state $x(0)$ (see~\cite{JadbabaieErgodic,scaglione}). In order to have $c = x_{\mbox{\scriptsize{avg}}}$, almost surely, an additional
condition must be satisfied:
 \begin{equation}
 \label{eqn_assymetric_condition}
 1^T \mathcal{W} (k)= 1^T,\,\,\mathrm{a.s.}
 \end{equation}
See~\cite{JadbabaieErgodic,scaglione} for the details.
We remark that~\eqref{eqn_assymetric_condition} is a crucial assumption in the derivation of the MSE decay~\eqref{eqn2_Lemma_MSE_k}.
 Theoretically, equation~\eqref{EqnRhoW} is still valid if the condition $\mathcal{W}(k)=\mathcal{W}(k)^T$ is relaxed to $1^T \mathcal{W}(k) = 1^T$.
While this condition is trivially satisfied for symmetric links and symmetric weights $W_{ij}=W_{ji}$, it is very difficult to
 realize~\eqref{eqn_assymetric_condition} in practice when the random links are asymmetric. So, in our work, we do not assume~\eqref{eqn_assymetric_condition} with asymmetric links.

For asymmetric networks, we follow reference~\cite{scaglione} and introduce
 the mean square state deviation $\mathrm{MSdev}$ as a performance measure. Denote the current
  average of the node states by $x_{\mbox{\scriptsize{avg}}}(k) = \frac{1}{N}1^Tx(k)$. Quantity $\mathrm{MSdev}$ describes how far apart different
  states $x_i(k)$ are; it is given by
  \[
  \mathrm{MSdev}(k) =  \sum_{i=1}^N \mathrm{E} \left[ ( x_i(k) - x_{\mbox{\scriptsize{avg}}}(k) )^2 \right] = \mathrm{E} \left[ \zeta(k)^T \zeta(k) \right],
  \]
  where
  \begin{equation}
  \zeta(k) = x(k) - x_{\mbox{\scriptsize{avg}}}(k)1 = (I - J)x(k).
  \end{equation}
\subsection{Symmetric links: Statistics of $\mathcal{W}(k)$}
In this subsection, we derive closed form expressions for the first and the second order statistics
 on the random matrix $\mathcal{W}(k)$. 
Let $q(k)$ be the random vector that collects the non redundant entries of $\mathcal{A}(k)$:
\begin{equation}
\label{eqn_definition_q}
q_l(k) = \mathcal{A}_{ij}(k),\, i<j,\, (i,j) \in E,
 \end{equation}
 where the entries of $\mathcal{A}(k)$ are ordered in lexicographic order with respect to $i$ and $j$, from left to right,
 top to bottom.
 For symmetric links, $\mathcal{A}_{ij}(k)=\mathcal{A}_{ji}(k)$, so the dimension of $q(k)$ is
 half of the number of directed links, i.e., $M$.
 We let the mean and the covariance of $q(k)$ and $\mathrm{Vec} \left( \mathcal{A} (k) \right)$ be:
\begin{eqnarray}
\label{eqn_pi_definition}
\pi &=& \mathrm{E}\left[ q(k)\right]\\
{\pi}_l&=& \mathrm{E} [q_l(k)]\\
\label{eqn_R_q_definition}
 R_q &=& \mathrm{Cov}( q(k) )  =
\mathrm{E}[\, ( q(k)  -  \pi)  \,\,(
q(k)  -  \pi) ^T\,] \\
R_A &=& \mathrm{Cov}(\, \mathrm{Vec}(\mathcal{A}(k)) \,)
\end{eqnarray}
The relation between $R_q$ and $R_A$ can be written as:
\begin{equation}
R_A = F R_q  F^T
\label{R_A}
\end{equation}
where $F \in {\mathbb R}^{N^2  \times M}$ is the zero one selection matrix that linearly maps $q(k)$ to $\mathrm{Vec}\left( \mathcal{A} (k) \right)$,
i.e., $\mathrm{Vec}\left( {\mathcal{A}} (k) \right) = F q(k).$
 We introduce further notation. Let $P$ be the matrix of the link formation probabilities
 \[
 P=\left[ P_{ij} \right]
 \]
Define the matrix $B \in {\mathbb R}^{N^2 \times N^2}$ with $N\times N$ zero diagonal blocks and $N\times N$ off
diagonal blocks $B_{ij}$ equal to:
\[
B_{ij} = 1e_i^T + e_j 1^T
\]
and write $W$ in terms of its columns $W = \left[ W_1\,\,W_2\,\,...\,\,W_N \right]$.
We let
\begin{eqnarray*}
W_C &=& W_1 \oplus W_2 \oplus ... \oplus W_N
\end{eqnarray*}
For symmetric random networks, the mean of the random weight matrix $\mathcal{W}(k)$ and of $\mathcal{W}^2(k)$
 play an important role for the convergence rate of the consensus algorithm. Using the above notation, we can get compact representations
 for these quantities, as provided in Lemma 1 proved in Appendix A.
\begin{lemma}
\label{LemmaStatisticsOfW}
Consider the consensus algorithm~(\ref{Eqn_consensus_Komapktno}).
Then the mean and the second moment~$R_C$ of $\mathcal{W}$ defined below are:
\begin{eqnarray}
\label{eqn:meancalW}
\overline{\rule{0pt}{9pt}W} &=& \mathrm{E}\left[ \mathcal{W} \right] = W \odot {P} + I - \mathrm{diag} \left( W {P}\right)\\
\label{eqn:covarcalW}
R_C&=& \mathrm{E}\left[ \mathcal{W}^2 \right] - \overline{\rule{0pt}{9pt}W}^2\\
&=& {W_C}^T\,\,\left\{\,{R_A}\odot (\, I\otimes  11^T\,\,\,        + \, 11^T \otimes  I \,\,\, - B)\right\}\,\,{W_C}
\end{eqnarray}
In the special case of spatially uncorrelated links, the second moment~$R_C$ of~$\mathcal{W}$ are
\begin{equation}
\label{eqn_R_C_uncorrelated}
\frac{1}{2}R_C =  \mathrm{diag}\left \{ \left( {\left(11^T - P \right)} \odot P \right) \left( W \odot W \right)\right \} -  {\left(11^T - P \right)} \odot P  \odot  W \odot W
\end{equation}
\end{lemma}
For asymmetric random links,
the expression for the mean of the random weight matrix $\mathcal{W}(k)$
 remains the same (as in Lemma 1).
  For asymmetric random links, instead of $\mathrm{E}\left[ \mathcal{W}^2(k)\right]-J $ (consider eqn.~(18),(19) and the term
  $\mathrm{E}\left[ \mathcal{W}^2(k)\right]$ in it), the
  quantity of interest becomes $\mathrm{E} \left[  \mathcal{W}^T \left( I-J \right) \mathcal{W}(k)  \right]$ (The quantity
 of interest is different since the optimization criterion will be different.)
  For symmetric links, the matrix $\mathrm{E} \left[  \mathcal{W}^2\right]-J$
  is a quadratic matrix function of the weights $W_{ij}$; it depends also quadratically on the $P_{ij}$'s
  and is an affine function of $\left[ R_q \right]_{ij}$'s. The same will still hold for
   $\mathrm{E} \left[  \mathcal{W}^T \left( I-J \right) \mathcal{W}(k)  \right]$
 in the case of asymmetric random links. The difference, however, is that
     $\mathrm{E} \left[  \mathcal{W}^T \left( I-J \right) \mathcal{W}(k)  \right]$
      does not admit the compact representation as given in~\eqref{eqn:covarcalW}, and we do not
       pursue here cumbersome entry wise representations. In the Appendix C,
         we do present the expressions for the matrix $\mathrm{E} \left[  \mathcal{W}^T \left( I-J \right) \mathcal{W}(k)  \right]$
      for the broadcast gossip algorithm~\cite{scaglione} (that we study in subsection~\ref{subsect_broadcast_gossip}).

\section{Weight optimization: symmetric random links}
\label{weight_optimization_symmetric}
\subsection{Optimization criterion: Mean square convergence rate}
\label{subsect_mss_rate}

We are interested in finding the rate at which $\mathrm{MSE}(k)$ decays to zero and to optimize
 this rate with respect to the weights $W$.
First we derive the recursion for the error $e(k)$. We have from eqn.~(\ref{Eqn_consensus_Komapktno}):
\begin{eqnarray*}
1^Tx(k+1) &=&
1^T \, \mathcal{W}(k)x(k)=1^Tx(k)=1^Tx(0)=N\,x_{\mbox{\scriptsize{avg}}}\\
1^Te(k) &=& 1^Tx(k) - 1^T\,1\,x_{\mbox{\scriptsize{avg}}} =
N\,x_{\mbox{\scriptsize{avg}}} - N\,x_{\mbox{\scriptsize{avg}}} = 0
\end{eqnarray*}
We derive the error vector dynamics:
\begin{equation}
e(k+1) = x(k+1) - x_{\mbox{\scriptsize{avg}}}\,1 
= \mathcal{W}(k)x(k) -
\mathcal{W}(k)\,x_{\mbox{\scriptsize{avg}}}\,1 
=  \mathcal{W}(k)e(k) 
= \left( \mathcal{W}(k) - J \right) \,e(k)
\label{eqn_error_dynamics}
 \end{equation}
where the last equality holds because $Je(k) = \frac{1}{N}11^Te(k) = 0$.

Recall the definition of the mean squared consensus error~\eqref{eqn_MSE} and the error covariance matrix
in eqn.~\eqref{eqn_Sigma_k} and recall that $\mathrm{MSE}(k)=\mathrm{tr}\, \Sigma(k)=\mathrm{E} \left[ e(k)e(k)^T\right]$.
Introduce the quantity
\begin{equation}
\phi(W) = \lambda_{\mathrm{max}}\left(  \mathrm{E}[\mathcal{W}^2]-J  \right)
\label{EqnRhoW}
\end{equation}
The next Lemma shows that the mean squared error decays at the rate $\phi(W)$.
\begin{lemma}[m.s.s convergence rate]
\label{Lemma_MSE_K}
Consider the consensus algorithm given by eqn.~(\ref{Eqn_consensus_Komapktno}). Then:
\begin{eqnarray}
\label{eqn1_Lemma_MSE_k}
\mathrm{tr}\left( \Sigma(k+1) \right) &=& \mathrm{tr}\left( \left(\mathrm{E}[\mathcal{W}^2]-J\right)\Sigma(k) \right)\\
\label{eqn2_Lemma_MSE_k}
\mathrm{tr}\left( \Sigma(k+1) \right) &\leq&  \left(\phi(W)\right) \,\, \mathrm{tr}\left ( \Sigma(k) \right ),\,k\geq 0
\end{eqnarray}
\end{lemma}
\begin{proof}
From the definition of the covariance $\Sigma(k+1)$, using the dynamics of the error $e(k+1)$, interchanging expectation with the $\mathrm{tr}$ operator, using properties of the trace,  interchanging the expectation with the $\mathrm{tr}$ once again, using the independence of $e(k)$ and $\mathcal{W}(k)$, and, finally, noting that $\mathcal{W}(k) J=J$, we get~(\ref{eqn1_Lemma_MSE_k}). The independence between $e(k)$ and $\mathcal{W}(k)$ follows because $\mathcal{W}(k)$ is an i.i.d. sequence, and $e(k)$ depends on $\mathcal{W}(0)$,..., $W(k-1)$. Then $e(k)$ and $\mathcal{W}(k)$ are independent by the disjoint block theorem~\cite{Kar}. Having~\eqref{eqn1_Lemma_MSE_k}, eqn.~\eqref{eqn2_Lemma_MSE_k} can be easily shown, for example, by exercise 18, page 423,~\cite{MatrixAnalysis}.
\end{proof}
We remark that, in the case of asymmetric random links, MSE does not asymptotically go to zero. For the case of
asymmetric links, we use different performance metric. This will be detailed in section~\ref{weight_optimization_asymmetric}.
%
%
%
\subsection{Symmetric links: Weight optimization problem formulation}
\label{Problem_Formulation}
We now formulate
the weight optimization problem as finding the weights $W_{ij}$ that optimize the mean squared
 rate of convergence:
\begin{equation}
\begin{array}[+]{ll}
\mbox{minimize} & \phi(W) \\
\mbox{subject to} & W \in S_W
\end{array}
\label{Eqn_Our_optimization_problem}
\end{equation}
The set $S_W$ is defined in eqn.~(\ref{eqn_S_W_symmetric}) and the rate $\phi(W)$ is given by~(\ref{EqnRhoW}).
The optimization problem~(\ref{Eqn_Our_optimization_problem}) is unconstrained, since effectively the optimization variables are $W_{ij} \in {\mathbb R}$, $(i,j) \in E$, other entries of $W$ being zero.

A point $W^{\bullet} \in S_{W}$ such that $\phi(W^{\bullet})<1$ will always exist if the supergraph $G$ is connected.
Reference~\cite{jadbabaie_on_consensus} studies the case when the random matrices $\mathcal{W}(k)$ are stochastic and shows that $\phi(W^{\bullet})<1$ if the supergraph is connected and all the realizations of the random matrix
 $\mathcal{W}(k)$ are stochastic symmetric matrices. Thus, to locate a point $W^{\bullet} \in S_W$ such that $\phi(W^{\bullet})<1$, we just take $W^{\bullet}$ that assures all the realizations of $\mathcal{W}$ be symmetric stochastic matrices. It is trivial to show that for any point in the set
\begin{equation}
\label{eqn_S_stoch}
S_{\mathrm{stoch}} = \{W\in S_W:\,\, W_{ij}>0,\,\, \mathrm{if}\,\,(i,j) \in \mathrm{E},\,\, W1<1 \} \subseteq S_W
\end{equation}
all the realizations of $\mathcal{W}(k)$ are stochastic, symmetric. Thus, for any point $W^{\bullet}  \in  S_{\mathrm{stoch}}$, we have that
$\phi(W^{\bullet})<1$ if the graph is connected.

%
%
%
%
We remark that the optimum $W^{*}$ does not have to lie in the set $S_{\mathrm{stoch}}$. In general, $W^{*}$ lies in the set
\begin{equation}
\label{eqn_S_conv}
S_{\mathrm{conv}} = \left\{ W \in S_W: \phi(W)<1\right\} \subseteq S_W
\end{equation}
The set $S_{\mathrm{stoch}}$ is a proper subset of $S_{\mathrm{conv}}$ (If $W \in S_{\mathrm{stoch}}$ then $\phi(W)<1$, but the converse statement is not true in general.) We also remark that the consensus algorithm~\eqref{Eqn_consensus_Komapktno} converges \emph{almost surely} if
$\phi(W)<1$ (not only in mean
 squared sense). 
This can be shown, for instance, by the technique developed in~\cite{jadbabaie_on_consensus}.

We now relate~(\ref{Eqn_Our_optimization_problem}) to reference~\cite{BoydWeights}. This reference studies the weight optimization for the case of a $static$ topology. In this case the topology is deterministic, described by the supergraph $G$. The link formation probability matrix $P$ reduces to the supergraph adjacency (zero-one) matrix~$A$, since the links occur always if they are realizable. Also, the link covariance matrix $R_q$ becomes zero.
The weight matrix $\mathcal{W}$ is deterministic and equal to
\begin{eqnarray*}
\mathcal{W}& = & \overline{\rule{0pt}{9pt}W} = 	\mathrm{diag}\left(	W A	\right) - W \odot A	+I
\end{eqnarray*}
Recall that $r(X)$ denotes the spectral radius of $X$. Then, the quantities $\left(r \left(    	 \mathcal{W}-J	 \right)\right)^2$ and $\phi\left( W \right)$ coincide. Thus, for the case of static topology, the optimization problem~(\ref{Eqn_Our_optimization_problem}) that we address reduces to the optimization problem proposed in~\cite{BoydWeights}.
\subsection{Convexity of the weight optimization problem}
\label{propertiesweightopt}
We show that  $\phi:\,S_W \rightarrow {\mathbb R}_{+}$ is convex, where $S_W$ is defined in eqn.~(\ref{eqn_S_W_symmetric}) and $\phi(W)$ by eqn.~(\ref{EqnRhoW}).

Lemma~\ref{LemmaStatisticsOfW} gives the closed form expression of $\mathrm{E}\left[\mathcal{W}^2\right]$. We see that $\phi(W)$ is the concatenation of a quadratic matrix function and $\lambda_{\mathrm{max}}(\cdot)$. This concatenation is not convex in general. However, the next Lemma shows that $\phi(W)$ is convex for our problem.
\begin{lemma}[Convexity of $\phi(W)$]
\label{Lemma_convexity}
The function $\phi:\,S_W \rightarrow {\mathbb R}_{+}$ is convex.
\end{lemma}
\begin{proof}
Choose arbitrary $X,\,Y \in S_W$. We restrict our attention to matrices $W$ of the form
\begin{equation}
\label{eqn_W_X_Y}
W = X + t \,Y,\,t \in {\mathbb R}.
\end{equation}
Recall the expression for $\mathcal{W}$ given by~(\ref{Eqn_PBW}) and~(\ref{Eqn_consensus_Komapktno}).
For the matrix $W$ given by~(\ref{eqn_W_X_Y}), we have for $\mathcal{W}=\mathcal{W}(t)$
\begin{eqnarray}
\label{eqn_W_cal_X_cal_Y_cal}
\mathcal{W}(t) &=& I - \mathrm{diag} \left[  \left( X+tY\right)\,\mathcal{A} \right] + \left( X+tY\right)\odot \mathcal{A}\\
&=& \mathcal{X} + t \mathcal{Y}, \: \nonumber
\mathcal{X} = X \odot \mathcal{A} + I - \mathrm{diag} \left( X \mathcal{A} \right),\:
\mathcal{Y} = Y \odot \mathcal{A}  - \mathrm{diag} \left( X \mathcal{A} \right)
\end{eqnarray}
Introduce the auxiliary function $\eta : \,{\mathbb R}\rightarrow {\mathbb R_{+}},$
\begin{equation}
\label{etafcn}
\eta(t) = \lambda_{\mathrm{max}} \left(  \mathrm{E} \left[ \mathcal{W}(t)^2\right] - J    \right)
\end{equation}
To prove that $\phi(W)$ is convex, it suffices to prove that the function $\phi$ is convex. Introduce $\mathcal{Z}(t)$ and compute successively
\begin{eqnarray}
\label{eqn_w_sq_J}
\mathcal{Z}(t) &=& \mathcal{W}(t)^2 - J  \\
& = & \left( \mathcal{X} + t \mathcal{Y} \right)^2-J \\
& = & t^2 \, \mathcal{Y}^2  + t \, \left(\mathcal{X} \mathcal{Y} +  \mathcal{Y}\mathcal{X}\right)  +
 \mathcal{X}^2 - J \\
 \label{eqn_z_2_z_1_z_0}
 &=&  t^2 \,  \mathcal{Z}_2 + t \, \mathcal{Z}_1 + \mathcal{Z}_0
 \end{eqnarray}
The random matrices $\mathcal{Z}_2$, $\mathcal{Z}_1$, and $\mathcal{Z}_0$ do not depend on $t$. Also, $\mathcal{Z}_2$ is semidefinite positive. The function $\eta(t)$ can be expressed as
$$\eta(t) = \lambda_{\mathrm{max}} \left( \mathrm{E} \left[\mathcal{Z}(t)\right]\right)$$
We will now derive that
\begin{equation}
\mathcal{Z}\left( (1-\alpha) t + \alpha u\right)   \preceq  (1-\alpha) \, \mathcal{Z}(t) + \alpha \, \mathcal{Z} \left(  u \right),  \,\, \: \forall \alpha \in \left[0,1\right], \: \forall t,u \in {\mathbb R}
\label{Z_inequality}
\end{equation}
Since $\eta(t) = t^2$ is convex, the following inequality holds:
\begin{equation}
  \left[ (1-\alpha) t + \alpha u\right]^2 \leq (1-\alpha) t^2 + \alpha u^2,\,\: \alpha \in \left[0,1\right]
  \label{eq_quadratic}
\end{equation}
Since the matrix $\mathcal{Z}_2$ is positive semidefinite, eqn.~(\ref{eq_quadratic}) implies that:
$$  \left( \left( (1-\alpha) t + \alpha u\right)^2 \right) \mathcal{Z}_2
 \preceq   (1-\alpha)\, t^2 \,\mathcal{Z}_2 + \alpha \,u^2 \, \mathcal{Z}_2  ,\,\,\alpha \in \left[0,1\right] $$
After adding to both sides $\left( (1-\alpha) t + \alpha u\right) \, \mathcal{Z}_1 + \mathcal{Z}_0$, we get eqn.~(\ref{Z_inequality}).
Taking the expectation to both sides of~(\ref{Z_inequality}), get:
\begin{eqnarray*}
  \mathrm{E}  \left[ \, \mathcal{Z} \left( (1-\alpha) t + \alpha u \right) \,\right]
 & \preceq  &
  \mathrm{E}  \left[\, (1-\alpha) \mathcal{Z}(t) + \alpha \mathcal{Z}\left(  u  \right) \,\right] \\
& = & (1-\alpha) \mathrm{E}  \left[\,  \mathcal{Z}\left( t  \right) \, \right] + \alpha \mathrm{E}  \left[\,  \mathcal{Z} \left( u  \right) \, \right],\,\,\alpha \in \left[0,1\right]
\end{eqnarray*}
Now, we have that:
\begin{eqnarray*}
 \eta \left( (1-\alpha)t + \alpha u \right)
 & = & \lambda_{\mathrm{max}} \left( \, \mathrm{E} \left[  \mathcal{Z} \left( (1-\alpha) t + \alpha u \right)  \right]  \, \right) \\
& \leq & \lambda_{\mathrm{max}} \left(\,  (1-\alpha) \mathrm{E}  \left[\mathcal{Z}(t)\right] + \alpha \mathrm{E} \left[ \mathcal{Z} \left(  u  \right)  \right] \, \right)
 \\
& \leq &  (1-\alpha) \, \lambda_{\mathrm{max}} \left(\,   \mathrm{E}  \left[\mathcal{Z}(t)\right]  \,\right)
+ \alpha \, \lambda_{\mathrm{max}} \left( \,  \mathrm{E}  \left[ \mathcal{Z} \left(  u  \right)\right]  \, \right)\\
& = &  (1-\alpha) \, \eta(t) + \alpha \,\eta(u),\,\,\alpha \in \left[0,1\right]
\end{eqnarray*}
The last inequality holds since $\lambda_{\mathrm{max}}(\cdot)$ is convex. This implies  $\eta(t)$ is convex and hence $\phi(W)$ is convex.
\end{proof}
We remark that convexity of $\phi(W)$ is not obvious and requires proof. The function $\phi(W)$ is a concatenation of a matrix
 quadratic function and $\lambda_{\mathrm{max}}(\cdot)$. Although the function $\lambda_{\mathrm{max}}(\cdot) $ is a convex function of
 \emph{its argument}, one still have to show that the following concatenation is convex: $W \mapsto \mathrm{E}[\mathcal{W}^2]-J \mapsto
 \phi(W) = \lambda_{\mathrm{max}} \left( \mathrm{E}[\mathcal{W}^2]-J \right)$.
\subsection{Fully connected random network: Closed form solution}
\label{analytical_study}
To get some insight how the optimal weights depend on the network parameters, we consider the impractical, but simple
 geometry of a complete random symmetric graph. For this example, the optimization problem~\eqref{Eqn_Our_optimization_problem} admits a closed form solution, while, in general, numerical optimization is needed to solve~\eqref{Eqn_Our_optimization_problem}.
Although not practical, this example provides insight how the optimal weights depend on the network size $N$, the link formation probabilities, and the link formation spatial correlations.
 The supergraph is symmetric, fully connected, with $N$ nodes and $M=N(N-1)/2$ undirected links. We assume that all the links have the same formation probability, i.e., that $\mathrm{Prob} \left( q_l = 1 \right) = \pi_l =p$, $p \in (0,1]$, $l=1,...,M$. We assume that the cross-variance between any pair of links $i$ and $j$ equals to $\left[ R_q \right]_{ij} = \beta \, p(1-p)$, where $\beta$ is the correlation coefficient. The matrix $R_q$ is given by
\[
R_q = p(1-p) \left[ (1-\beta) I + \beta \, 11^T \right].
\]
The eigenvalues of $R_q$ are $\lambda_1(R_q) = p(1-p) \left(1 + (M-1) \, \beta\right)$, and $\lambda_i(R_q ) = p(1-p)\left(1 - \beta\right) \geq 0$, $i=2,...,M$.
The condition that $R_q \succeq 0$ implies that $\beta \geq -1/(M-1)$. Also, we have that
\begin{eqnarray}
\beta &:=&\frac { \mathrm{E} \left[  q_i q_j  \right] - \mathrm{E} \left[  q_i\right] \mathrm{E} \left[ q_j  \right]  }
  {   \sqrt{ \mathrm{Var}(q_i) }  \sqrt{ \mathrm{Var}(q_j)}  }\\
 &=& \frac { \mathrm{Prob} \left(  q_i=1, q_j=1  \right) - p^2}
  {p(1-p)} \geq -\frac{p}{1-p}
\end{eqnarray}
Thus, the range of $\beta$ is restricted to
  \begin{equation}
  \label{beta_range}
  \mathrm{max} \left( \frac{-1}{M-1}, \frac{-p}{1-p} \right)\leq \beta \leq 1.
  \end{equation}

Due to the problem symmetry, the optimal weights for all links are the same, say $W^{*}$. The expressions for the optimal weight $W^{*}$ and for the optimal convergence rate $\phi^{*}$ can be obtained after careful manipulations and expressing the matrix $\mathrm{E} \left[  \mathcal{W}^2 \right]-J$ explicitly in terms of $p$ and $\beta$; then, it is easy to show that:
\begin{eqnarray}
\label{eqn_optimal_weights}
W^{*} &=& \frac{1} {Np + (1-p) \left( 2 + \beta(N-2) \right)}\\
\label{eqn_optimal_rate}
\phi^{*} &=& 1 - \frac{1}{1 + \frac{1-p}{p} \left( \frac{2}{N}(1-\beta) + \beta  \right) }
\end{eqnarray}

The optimal weight $W^{*}$ decreases as $\beta$ increases. This is also intuitive, since positive correlations imply that the links emanating from the same node tend to occur simultaneously, and thus the weight should be smaller. Similarly, negative correlations imply that the links emanating from the same node tend to occur exclusively, which results in larger weights. Finally, we observe that in the uncorrelated case ($\beta = 0$), as $N$ becomes very large, the optimal weight behaves as $1/(Np)$. Thus, for the uncorrelated links and large network, the optimal strategy (at least for this example) is to rescale the supergraph-optimal weight $1/N$ by its
formation probability $p$. Finally, for fixed $p$ and $N$, the fastest rate is achieved when $\beta$ is as negative as possible.

\subsection{Numerical optimization: subgradient algorithm}
\label{Subgradient}
We solve the optimization problem in~(\ref{Eqn_Our_optimization_problem}) for generic networks by the subgradient algorithm,~\cite{Urruty}.
In this subsection, we consider spatially uncorrelated links, and we comment on extensions for spatially correlated links. Expressions for spatially correlated links are provided in Appendix B.

We recall that the function $\phi({W})$ is convex (proved in Section~\ref{propertiesweightopt}). It is nonsmooth because ${\lambda}_{\mathrm{max}}(\cdot)$ is nonsmooth. Let $H \in {\mathbb S}^N$ be the subgradient of the function $\phi(W)$. To derive the expression for the subgradient of $\phi(W)$,
 we use the variational interpretation of $\phi(W)$:
\begin{eqnarray}
\phi(W) &=&  \max_{v^Tv=1}{v^T\,\left( \mathrm{E}\left[ \mathcal{W}^2\right] -J \right) \,v} 
\label{eqn_subgrad_calc}
= \max_{v^Tv=1}{f_v(W)}
\end{eqnarray}
By the subgradient calculus, a subgradient of $\phi(W)$ at point $W$ is equal to a subgradient $H_{u}$ of the function $f_u(W)$ for which the maximum of the optimization problem~(\ref{eqn_subgrad_calc}) is attained, see, e.g.,~\cite{Urruty}. The maximum of $f_v(W)$ (with respect to $v$)
 is attained at $v=u$, where $u$ is the eigenvector of the matrix $\mathrm{E}\left[ \mathcal{W}^2\right] -J$ that corresponds to its maximal eigenvalue, i.e., the maximal eigenvector. In our case, the function $f_u(W)$ is differentiable (quadratic function), and hence the subgradient of $f_u(W)$ (and also the subgradient of $\phi(W)$) is equal to the gradient of $f_u(W)$,~\cite{Urruty}:
\begin{equation}
  H_{ij} = \left\{ \begin{array}{rl}
  u^T\,\frac{\partial \left(   \mathrm{E}\left[ \mathcal{W}^2\right] -J       \right)}{\partial W_{ij}}\,u &\mbox{if $(i,j) \in E$ }\\
  0 &\mbox{otherwise.}
       \end{array} \right.
       \label{eq:PBW}
\end{equation}
%
%
We compute for $(i,j) \in E$
\begin{eqnarray}
H_{ij} & = & u^T\,\frac{\partial \left(  \overline{\rule{0pt}{9pt}W}^2-J + R_C     \right)}{\partial W_{ij}}\,u
\label{eq:PBW-b}\\
& = &  u^T\ \left( - 2\,\overline{\rule{0pt}{9pt}W}\,P_{ij}(e_i-e_j)(e_i-e_j)^T + 4\,W_{ij}\, P_{ij}(1-P_{ij})(e_i-e_j)(e_i-e_j)^T  \right)\,u \nonumber \\
& = &2  {P}_{ij} (u_i - u_j)u^T(\overline{\rule{0pt}{9pt}W}_j-\overline{\rule{0pt}{9pt}W}_i) + 4{P}_{ij}(1-P_{ij})W_{ij}(u_i-u_j)^2
\end{eqnarray}
\begin{algorithm}
\caption{Subgradient algorithm}
\begin{algorithmic}
    \STATE  Set initial $W^{(1)} \in S_W$ \\
    \STATE  Set $k=1$  \\
    \STATE
    Repeat
        \\ $\,\,\,\,\,\,\,\,\,\,$ Compute a subgradient $H^{(k)}$ of $\phi$ at $W^{(k)}$, and set $W^{(k+1)} = W^{(k)}-{\alpha}_kH^{(k)}$
        \\ $\,\,\,\,\,\,\,\,\,\,$ $k:=k+1$
\end{algorithmic}
\label{subradient-algorithm}
\end{algorithm}
The subgradient algorithm is given by algorithm~\ref{subradient-algorithm}. The stepsize ${\alpha}_k$ is nonnegative, diminishing, and nonsummable: ${\mathrm{lim}}_{k\rightarrow \infty} {\alpha}_k = 0$, $\sum_{k=1}^{\infty} {\alpha}_k = \infty$. We choose $\alpha_k = \frac{1}  {\sqrt{k}}$, $k=1,2,...$, similarly as in~\cite{BoydWeights}.
\section{Weight optimization: asymmetric random links}
\label{weight_optimization_asymmetric}
We now address the weight optimization for asymmetric random networks. Subsections~\ref{opt_criter_asymm} and V-B introduce the optimization criterion and the corresponding weight optimization problem, respectively.
Subsection V-C shows that this optimization problem is convex.
\subsection{Optimization criterion: Mean square deviation convergence rate}
\label{opt_criter_asymm}
  Introduce now
  \begin{equation}
  \label{PSI_FUNCTION}
  \psi(W): = \lambda_{\mathrm{max}} \left(  \mathrm{E} \left[ \mathcal{W}^T \left( I-J\right) \mathcal{W} \right] \right).
  \end{equation}
Reference~\cite{scaglione} shows that the mean square deviation $\mathrm{MSdev}$
 satisfies the following equation:
 \begin{equation}
 \label{eqn_zeta}
 \mathrm{MSdev}(k+1)  \leq  \psi(W) \, \mathrm{MSdev}(k).
 \end{equation}
 Thus, if the quantity $\psi(W)$ is strictly less than one, then $\mathrm{MSdev}$ converges to zero asymptotically, with the worst case
  rate equal to $\psi(W)$. We remark that the condition~\eqref{eqn_assymetric_condition} is not needed for eqn.~\eqref{eqn_zeta} to hold, i.e.,
  MSdev converges to zero even if condition~\eqref{eqn_assymetric_condition} is not satisfied; this condition is needed only for eqn.~\eqref{eqn2_Lemma_MSE_k} to hold, i.e., only to have MSE to converge to zero.
\subsection{Asymmetric network: Weight optimization problem formulation}
In the case of asymmetric links, we propose to optimize the mean square deviation convergence rate, i.e., to solve the following optimization problem:
\begin{equation}
\begin{array}[+]{ll}
\mbox{minimize} & \psi(W) \\
\mbox{subject to} & W \in S_W^{\mbox{\scriptsize{asym}}} \\
                & \sum_{i=1}^N P_{ij} W_{ij} =1,\,\,i=1,...,N
\end{array}
\label{Eqn_asym_optimization_problem}
\end{equation}
The constraints in the optimization problem~\eqref{Eqn_asym_optimization_problem} assure that, in expectation, condition~\eqref{eqn_assymetric_condition} is satisfied, i.e., that
\begin{equation}
\label{constraint}
1^T\,\mathrm{E} \left[ \mathcal{W}\right] = 1^T.
\end{equation}

If~\eqref{constraint} is satisfied, then the consensus algorithm converges to the true average $x_{avg}$ in expectation~\cite{scaglione}.

Equation~\eqref{constraint} is a linear constraint with respect to the weights $W_{ij}$, and thus does not violate the convexity of
the optimization problem~\eqref{Eqn_asym_optimization_problem}.
We emphasize that in the case of asymmetric links, we do not assume the weights $W_{ij}$ and $W_{ji}$ to be equal.
In section~\ref{subsect_broadcast_gossip}, we show that allowing $W_{ij}$ and $W_{ji}$ to be different leads to better
solutions in the case of asymmetric networks.
\subsection{Convexity of the weight optimization problem}
\label{propertiesweightopt_asym}
We show that the function $\psi(W)$ is convex.
We remark that reference~\cite{scaglione} shows that the function is convex, when all the weights $W_{ij}$ are equal
 to $g$. We show here that this function is convex even when the weights are different.
\begin{lemma}[Convexity of $\psi(W)$]
\label{Lemma_convexity_asym}
The function $\phi:\,S_W^{\mbox{\scriptsize{asym}}} \rightarrow {\mathbb R}_{+}$ is convex.
\end{lemma}
\begin{proof}
The proof is very similar to the proof of Lemma~\ref{Lemma_convexity}. The proof starts with introducing $W$ as in eqn.~\eqref{eqn_W_X_Y}
 and with introducing $\mathcal{W}(t)$ as in eqn.~\eqref{eqn_W_cal_X_cal_Y_cal}. The difference is that, instead of considering the matrix
 $\mathcal{W}^2-J$, we consider now the matrix $\mathcal{W}^T \,(I-J)\, \mathcal{W}$. In the proof of Lemma~\ref{Lemma_convexity}, we
  introduced the auxiliary function $\eta(t)$ given by~\eqref{etafcn}; here, we introduce the auxiliary function $\kappa(t)$, given by:
  \begin{equation}
  \kappa(t)  =  \lambda_{\mathrm{max}}  \left(    \mathcal{W}(t)^T (I-J) \mathcal{W}     \right),
  \end{equation}
  and show that $\psi(W)$ is convex by proving that $\kappa(t)$ is convex.
  Then, we proceed as in the proof of Lemma~\ref{Lemma_convexity}. In eqn.~\eqref{eqn_z_2_z_1_z_0}
    the matrix $\mathcal{Z}_2$ becomes $\mathcal{Z}_2:=\mathcal{Y}^T (I-J) \mathcal{Y}$.
   The random matrix $\mathcal{Z}_2$
   is obviously positive semidefinite. The proof then proceeds as in Lemma~\ref{Lemma_convexity}.
  \end{proof}
\section{Simulations}
\label{simulations}
We demonstrate the effectiveness of our approach with a comprehensive set of simulations. These simulations cover both examples of
asymmetric and symmetric networks and both networks with random link failures and with randomized protocols.
In particular, we consider the following two standard sets of experiments with random networks: 1) spatially correlated link failures and symmetric links and 2) randomized protocols, in particular, the broadcast gossip algorithm~\cite{scaglione}.
 With respect to the first set, we consider correlated link failures with two types of correlation structure.
 We are particularly interested in studying the
       dependence of the performance and of the gains
      on the size of the network $N$ and on the link correlation structure.

 In all these experiments, we consider geometric random graphs.  Nodes communicate among themselves
  if within their radius of communication, $r$. The nodes are uniformly distributed on a unit square. The number of nodes is
   $N=100$ and the average degree is $15\%N$. In subsection~\ref{unreliable_links}, the random instantiations of the networks are undirected; in subsection~VI-B, the random
  instantiations of the networks are directed.

  In the first set of experiments with correlated link failures, the link formation probabilities $P_{ij}$ are chosen such that they decay
   quadratically with the distance:
   \begin{equation}
   \label{eqn_P_ij_s}
   P_{ij} = 1-k\, \left(  \frac{\delta_{ij}}{r}  \right)^2,
   \end{equation}
where we choose $k=0.7$. We see that, with~\eqref{eqn_P_ij_s}, a link will be active with high probability
 if the nodes are close ($\delta_{ij} \simeq 0$), while the link will be down with
 probability at most $0.7$, if the nodes are apart by $r$.

%

 We recall that we refer to our weight design, i.e., to the solutions of the weight optimization problems~\eqref{Eqn_Our_optimization_problem},~\eqref{Eqn_asym_optimization_problem}, as probability based weights (PBW).
 We study the performance of PBW, comparing it with
  the standard weight choices available in the literature: in subsection~\ref{unreliable_links}, we compare it with the Metropolis weights (MW),
  discussed in~\cite{BoydWeights}, and the supergraph based weights (SGBW). The SGBW are
   the optimal (nonnegative) weights designed for a static (nonrandom) graph $G$,
    which are then applied to a random network when the underlying supergraph is $G$.
   This is the strategy used in~\cite{WhichShouldI}. For asymmetric links (and for asymmetric weights $W_{ij}\neq W_{ji}$), in subsection~\ref{subsect_broadcast_gossip}, we compare PBW
    with the optimal weight choice in~\cite{scaglione} for broadcast gossip that considers all the weights to be equal.

In the first set of experiments in subsection~\ref{unreliable_links}, we quantify the performance gain of PBW over SGBW and MW by the gains:
 \begin{equation}
 \label{gamma_s_tau}
     \Gamma_s^{\tau} = \frac{\tau_{\mathrm{SGBW}}}{\tau_{\mathrm{PBW}}}
     \end{equation}
  where $\tau$ is a time constant defined as:
  \begin{equation}
  \label{eqn_tau}
 \tau = \frac{1}{0.5\, \mathrm{ln}\, \phi(W)}
  \end{equation}
  We also compare PBW with SGBW and MW with the following measure:
  \begin{eqnarray}
  \label{eqn_gamma_s_m_eta}
  \Gamma_s^{\eta} &=&  \frac{\eta_{\mathrm{SGBW}}}{\eta_{\mathrm{PBW}}}    \\
  \Gamma_m^{\eta} &=&  \frac{\eta_{\mathrm{MW}}}{\eta_{\mathrm{PBW}}}
  \end{eqnarray}
  where $\eta$ is the asymptotic time constant defined by
  \begin{eqnarray}
  \label{eqn_eta}
  \eta &=& \frac{1}{|\gamma|}   \\
  \gamma &=& \lim_{k \rightarrow \infty}  \left( \frac{\|e(k)\|}{\|e(0)\|} \right)^{1/k}
  \end{eqnarray}
 Reference~\cite{WhichShouldI} shows that for random networks $\eta$ is an almost sure constant and $\tau$ is an upper bound on $\eta$.
  Also, it shows that
  $\tau$ is an upper bound on $\eta$.

Subsections~\ref{unreliable_links} and~\ref{subsect_broadcast_gossip} will provide further details on the expermints.
\subsection{Symmetric links: random networks with correlated link failures}
\label{unreliable_links}
To completely define the probability distribution of the random link vector $q \in {\mathbb R}^M$, we must assign probability to each of the $2^M$ possible realizations of $q$, $q = (\alpha_1,...,\alpha_M)^T$,  $\alpha_i \in \{0,1\}$. Since in networks of practical interest $M$ may be very large, of order $1000$ or larger, specifying the complete distribution of the vector $q$ is most likely infeasible. Hence, we work with the second moment description and specify only the first two moments of its distribution, the mean and the covariance, $\pi$ and $R_q$. Without loss of generality, order the links so that $\pi_1 \leq \pi_2 \leq...\leq \pi_M$.
\begin{lemma}
The mean and the variance $(\pi, R_q)$ of a Bernoulli random vector satisfy:
\begin{eqnarray}
\label{eqn_restriction_pi}
0 &\leq& \pi_i \, \, \leq
 \, \, 1,\,\,i=1,...,N\\
\label{eqn_restriction_Rq}
R_q &\succeq& 0 \\
\label{eqn_restriction_Rq_ij}
\mathrm{max} \left(    -\pi_i \pi_j,\, \pi_i + \pi_j -1 - \pi_i\,\pi_j    \right)   &\leq&  \left[ R_q \right]_{ij} \leq  \pi_i\,(1-\pi_j) = \overline{R}_{ij}, \,\, i < j
\end{eqnarray}
\end{lemma}
\begin{proof}
Equations~\eqref{eqn_restriction_pi} and~\eqref{eqn_restriction_Rq} must hold because $\pi_l$'s are probabilities and $R_q$
 is a covariance matrix.
Recall that
\begin{equation}
\label{eqn_R_q_ij_def}
\left[ R_q \right]_{ij} = \mathrm{E} \left[  q_i q_j  \right] - \mathrm{E} \left[  q_i  \right]\mathrm{E} \left[  q_j  \right]=
\mathrm{Prob} \left(  q_i=1,\, q_j=1  \right) - \pi_i \pi_j.
\end{equation}
To prove the lower bound in~\eqref{eqn_restriction_Rq_ij}, observe that:
\begin{eqnarray}
\mathrm{Prob} \left(  q_i=1,\, q_j=1  \right) &=& \mathrm{Prob} \left(  q_i=1 \right) + \mathrm{Prob} \left(   q_j=1  \right)
 - \mathrm{Prob} \left(  \{q_i=1\} \, \mathrm{or} \, \{q_j=1\}  \right)  \nonumber   \\
 &=& \pi_i + \pi_j - \mathrm{Prob} \left(  \{q_i=1\} \, \mathrm{or} \, \{q_j=1\}  \right)
 \geq \pi_i + \pi_j -1.
  \label{eqn_LB}
 \end{eqnarray}
In view of the fact that $\mathrm{Prob} \left(  q_i=1,\, q_j=1  \right) \geq 0$, eqn.~\eqref{eqn_LB},
and eqn.~\eqref{eqn_R_q_ij_def}, the proof for the lower bound in~\eqref{eqn_restriction_Rq_ij} follows.
The upper bound in~\eqref{eqn_restriction_Rq_ij} holds because $\mathrm{Prob} \left(  q_i=1,\, q_j=1  \right) \leq \pi_i,\,\,i<j$ and eqn.~\eqref{eqn_R_q_ij_def}.
\end{proof}
If we choose a pair $(\pi, R_q)$ that satisfies~\eqref{eqn_restriction_pi},~\eqref{eqn_restriction_Rq},~\eqref{eqn_restriction_Rq_ij}, one cannot guarantee that $\left( \pi, R_q \right)$ is a valid pair, in the sense that there exists a probability distribution on $q$ with its first and second moments being equal to $(\pi,R_q)$, ~\cite{Quadish}. Furthermore, if $(\pi,R_q)$ is given, to simulate binary random variables with the marginal probabilities and correlations equal to $(\pi,R_q)$ is challenging. These questions have been studied, see~\cite{Zucker,Quadish}. We use the results in~\cite{Zucker,Quadish} to generate our correlation models. In particular, we use the result that $\overline{R} = \left[ \overline{R}_{ij}\right]$ (see eqn.~\eqref{eqn_restriction_Rq_ij}) is a valid correlation structure for any $\pi$,~\cite{Zucker}. We simulate the correlated links by the method proposed in~\cite{Quadish}; this method handles a wide range of different correlation structures and has a small computational cost.

\mypar{Link correlation structures}
We consider two different correlation structures for any pair of links $i$ and $j$ in the supergraph:
\begin{eqnarray}
\label{eqn_first_corr_struc}
\left[R_q\right]_{ij} = c_1 \, \overline{R}_{ij}\\
\label{eqn_second_corr_struc}
\left[R_q\right]_{ij} = c_2 \, \theta^ {\kappa_{ij}} \, \overline{R}_{ij}
\end{eqnarray}
where $c_1 \in (0,1]$, $\theta \in (0,1)$ and $c_2 \in (0,1]$ are parameters, and $\kappa_{ij}$ is the distance between links $i$ and $j$ defined as the length of the shortest path that connects them in the supergraph.

The correlation structure~\eqref{eqn_first_corr_struc} assumes that the correlation between any pair of links is a fraction of the maximal possible correlation, for the given $\pi$ (see eqn.~\eqref{eqn_restriction_Rq_ij} to recall $\overline{R}_{ij}$). Reference~\cite{Zucker} constructs a method for generating the correlation structure~\eqref{eqn_first_corr_struc}.

The correlation structure~\eqref{eqn_second_corr_struc} assumes that the correlation between the links decays geometrically with this distance .  In our simulations, we set $\theta=0.95$, and find the maximal $c_2$, such that the resulting correlation structure can be simulated by the method in~\cite{Quadish}. For all the networks that we simulated in the paper, $c_2$ is between $0.09$ and $0.11$.

\mypar{Results} We want to address the following two questions: 1) What is the performance gain ($\Gamma_s$, $\Gamma_m$ in eqns.~\eqref{gamma_s},~\eqref{gamma_m}) of PBW over SGBW and MW;  and  2)  How does this gain scale with the network size, i.e., the number of nodes $N$?

\mypar{Performance gain of PBW over SGBW and MW}
We consider question 1) for both correlation structures~\eqref{eqn_first_corr_struc},~\eqref{eqn_second_corr_struc}. We generate $20$ instantiations of our standard supergraphs (with 100 nodes each and approximately the same average relative degree, equal to $15\%$). Then, for each supergraph, we generate formation probabilities according to rule~\eqref{eqn_P_ij_s}. For each supergraph with the given formation probabilities, we generate two link correlation structures,~\eqref{eqn_first_corr_struc} and~\eqref{eqn_second_corr_struc}.
We evaluate the convergence rate $\phi_j$ given by~\eqref{eqn2_Lemma_MSE_k}, time constants $\eta_j$ given by~\eqref{eqn_eta}, and $\tau_j$, given by~\eqref{eqn_tau}, and the performance gains $\left[\Gamma_s^{\eta}\right]_j$, $\left[\Gamma_m^{\eta}\right]_j$ for each supergraph ($j=1,...,20$). We compute the mean $\overline{\phi}$, the maximum $\phi^{+}$ and the minimum $\phi^{-}$ from the list $\{\phi_j\}$, $j=1,...,20$ (and similarly for $\{\eta_j\}$ and $\{\tau_j\}$, $j=1,...,20$).
Results for the correlation structure~\eqref{eqn_first_corr_struc} are given in Table 1 and for the correlation structure~\eqref{eqn_second_corr_struc}, in Table 2. The performance gains $\Gamma_s$, $\Gamma_m$, for both correlation structures are in Table 3. In addition, Figure 1 depicts the averaged error norm over 100 sample paths. We can see that the PBW outperform the SGBW and the MW
for both correlation structures~\eqref{eqn_first_corr_struc} and~\eqref{eqn_second_corr_struc}. For example, for the correlation~\eqref{eqn_first_corr_struc}, the PBW take less than 40 iterations to achieve $0.2\%$ precision, while the SGBW
 take more than 70, and the MW take more than 80 iterations. For correlation~\eqref{eqn_second_corr_struc}, to achieve $0.2\%$
  precision, the PBW take about 47 iterations, while the SGBW and the MW take more than 90 and 100 iterations, respectively.
\begin{figure}[thpb]
      \centering
      \includegraphics[height=2.2
      in,width=3.0in]{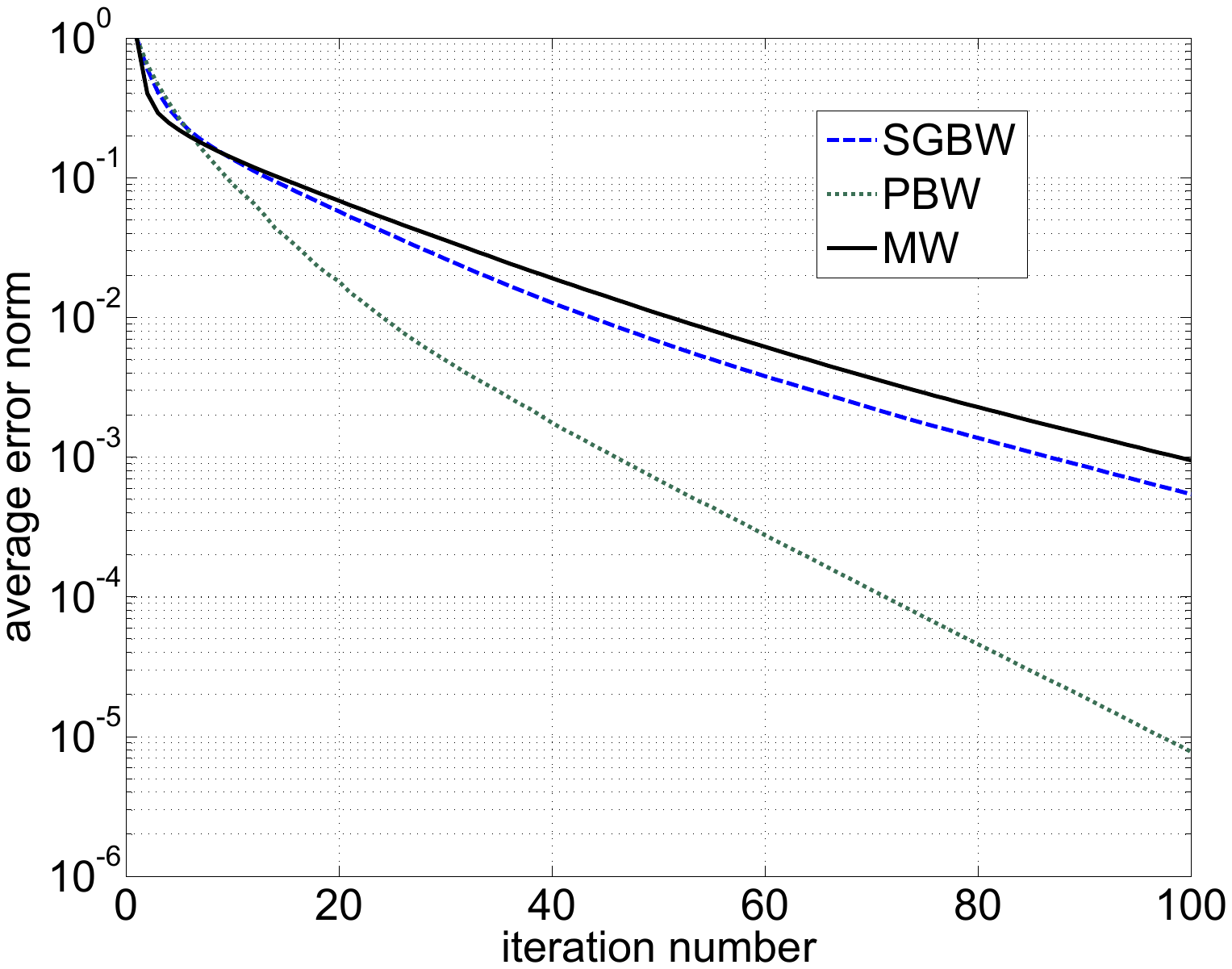}
      \includegraphics[height=2.2in,width=3.0in]{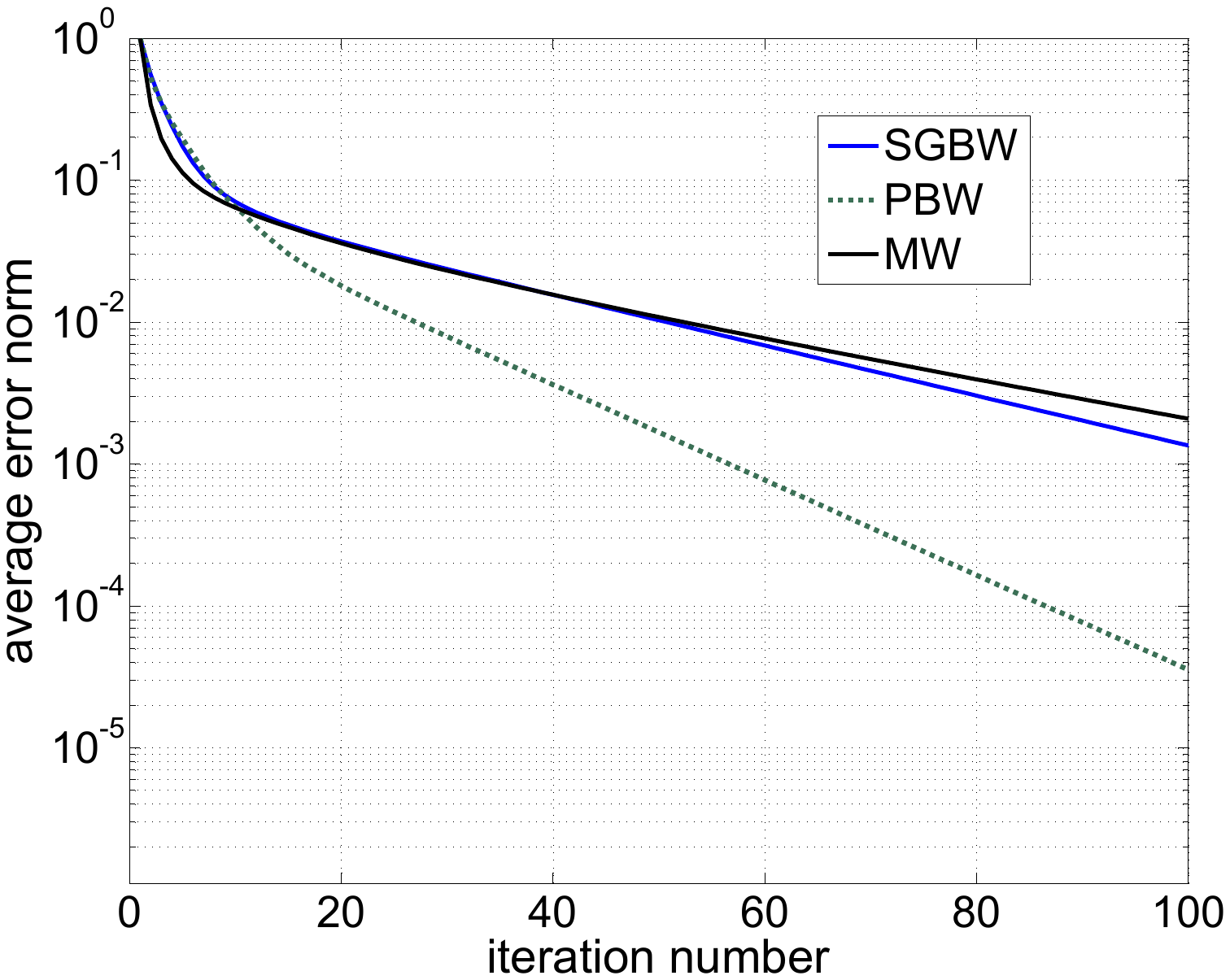}
      \label{figure_error_small_net_uncorrel}
      \caption{Average error norm versus iteration number.
      Left: correlation structure~\eqref{eqn_first_corr_struc}; right: correlation structure~\eqref{eqn_second_corr_struc}.}
      \label{Figure_Error_VS_k}
\end{figure}
\newcommand{\ra}[1]{\renewcommand{\arraystretch}{#1}}
\begin{table}
\begin{minipage}[b]{0.5\linewidth}
\centering
\caption{Correlation structure~\eqref{eqn_first_corr_struc}: Average $\overline{(\cdot)}$, maximal ${(\cdot)}^{+}$, and minimal $(\cdot)^{-}$
         values of the MSE convergence rate $\phi$~\eqref{EqnRhoW}, and corresponding time constants $\tau$~\eqref{eqn_tau} and $\eta$~\eqref{eqn_eta}, for 20 generated supergraphs}
         \ra{0.7}
\begin{small}
\begin{tabular}{@{}rrrrrrr@{}}\toprule
\, & SGBW & PBW & MW\\
\hline
\vspace*{-.cm}\\
$\overline{\phi}\,\,\,\,$ $\,\,\,\,\,$ & 0.91 & 0.87 &  \\
$\phi^{+}$ $\,\,\,\,\,$ & 0.95 & 0.92 &  \\
$\phi^{-}$ $\,\,\,\,\,$ & 0.89 & 0.83 &  \\
\hline
$\overline{\tau}\,\, \,\,$ $\,\,\,\,\,$& 22.7 & 15.4 &  \\
$\tau^{+} $ $\,\,\,\,\,$& 28 & 19 &  \\
$\tau^{-} $ $\,\,\,\,\,$&  20 & 14 & \\
\hline
$\overline{\eta} \,\, \,\, $ $\,\,\,\,\,$&    20   &   13   &   29\\
$\eta^{+} $ $\,\,\,\,\,$&           25   &   16  &    38\\
$\eta^{-} $ $\,\,\,\,\,$&           19   &   12   &   27\\
\hline
\bottomrule
\end{tabular}
\label{Table_1_Small_net}
\end{small}
\end{minipage}
\hspace{-0.99cm}
\begin{minipage}[b]{0.5\linewidth}
\centering
\caption{ Correlation structure~\eqref{eqn_second_corr_struc}: Average $\overline{(\cdot)}$, maximal ${(\cdot)}^{+}$, and minimal $(\cdot)^{-}$
         values of the MSE convergence rate $\phi$~\eqref{EqnRhoW}, and corresponding time constants $\tau$~\eqref{eqn_tau} and $\eta$~\eqref{eqn_eta}, for 20 generated supergraphs }
\ra{0.7}
\begin{small}
\begin{tabular}{@{}rrrrrrr@{}}\toprule
\, & SGBW & PBW & MW\\
\hline
\vspace*{-.06cm}\\
$\overline{\phi}\,\,\,\,$ $\,\,\,\,\,$ & 0.92 & 0.86 &  \\
$\phi^{+}$ $\,\,\,\,\,$ & 0.94 & 0.90 &  \\
$\phi^{-}$ $\,\,\,\,\,$ & 0.91 & 0.84 &  \\
\hline
$\overline{\tau}\,\, \,\,$ $\,\,\,\,\,$& 25.5 & 14.3 &  \\
$\tau^{+} $ $\,\,\,\,\,$& 34 & 19 &  \\
$\tau^{-} $ $\,\,\,\,\,$&  21 & 12 & \\
\hline
$\overline{\eta}\,\, \,\, $ $\,\,\,\,\,$&    20   &   11.5   &    24.4 \\
$\eta^{+} $ $\,\,\,\,\,$&                    23   &   14      &    29 \\
$\eta^{-} $ $\,\,\,\,\,$&                    16   &   9       &    19  \\
\hline
\bottomrule
\end{tabular}
\label{Table_2}
\end{small}
\end{minipage}
\hspace{0.1cm}
\begin{minipage}[b]{0.5\linewidth}
\centering
\caption{Average $\overline{(\cdot)}$, maximal $({\cdot})^{+}$, and minimal $(\cdot)^{-}$
          performance gains $\Gamma_s^{\eta}$ and $\Gamma_m^{\eta}$~\eqref{eqn_gamma_s_m_eta}
 for the two correlation structures~\eqref{eqn_first_corr_struc} and~\eqref{eqn_second_corr_struc} for 20 generated supergraphs}
\ra{0.7}
\begin{small}
\begin{tabular}{@{}rrrrrrr@{}}\toprule
\, & Correlation~\eqref{eqn_first_corr_struc} & Correlation~\eqref{eqn_second_corr_struc} \\
\hline
\vspace*{-.06cm}\\
$\overline{({\Gamma}_s^{\eta})}$ $\,\,\,\,\,$& 1.54 &
1.73  \\
$(\Gamma_s^{\eta})^{+} $ $\,\,\,\,\,$& 1.66 & 1.91 \\
$(\Gamma_s^{\eta})^{-} $ $\,\,\,\,\,$&  1.46 & 1.58 \\
\hline
$\overline{({\Gamma}_m^{eta})}$ $\,\,\,\,\,$ & 2.22 & 2.11  \\
$(\Gamma_m^{\eta})^{+}$ $\,\,\,\,\,$ & 2.42 & 2.45 \\
$(\Gamma_m^{\eta})^{-}$ $\,\,\,\,\,$ & 2.07 & 1.92 \\
\hline
\hline
\bottomrule
\end{tabular}
\label{Table_2}
\end{small}
\end{minipage}
\end{table}
%
The average performance gain of PBW over MW is larger than the performance gain over SGBW, for both~\eqref{eqn_first_corr_struc} and~\eqref{eqn_second_corr_struc}. The gain over SGBW, $\Gamma_s$, is significant, being 1.54 for~\eqref{eqn_first_corr_struc} and 1.73 for~\eqref{eqn_second_corr_struc}.
The gain with the correlation structure~\eqref{eqn_second_corr_struc} is larger than the gain with~\eqref{eqn_first_corr_struc}, suggesting that larger gain over SGBW is achieved with smaller correlations. This is intuitive, since large positive correlations imply that the random links tend to occur simultaneously, i.e., in a certain sense random network realizations are more similar to the underlying supergraph.

Notice that the networks with $R_q$ as in~\eqref{eqn_second_corr_struc} achieve faster rate than for~\eqref{eqn_first_corr_struc} (having at the same time similar supergraphs and formation probabilities). This is in accordance with the analytical studies in section~\ref{analytical_study} that suggest that faster rates can be achieved for smaller (or negative correlations) if $G$ and $\pi$ are fixed.

\mypar{Performance gain of PBW over SGBW as a function of the network size}
To answer question 2), we generate the supergraphs with $N$ ranging from $30$ up to $160$, keeping the average relative degree of the supergraph approximately the same ($15\%$). Again, PBW performs better than MW ($\tau_{\mathrm{SGBW}} < 0.85 \tau_{\mathrm{MW}}$), so we focus on the dependence of $\Gamma_s$ on $N$, since it is more critical.

Figure~\ref{figure_SGBW_PBW} plots $\Gamma_s$ versus $N$, for the two correlation structures.
The gain $\Gamma_s$ increases with $N$ for both~\eqref{eqn_second_corr_struc} and~\eqref{eqn_first_corr_struc}.
\begin{figure}[thpb]
      \centering
      \includegraphics[scale=0.25 ]{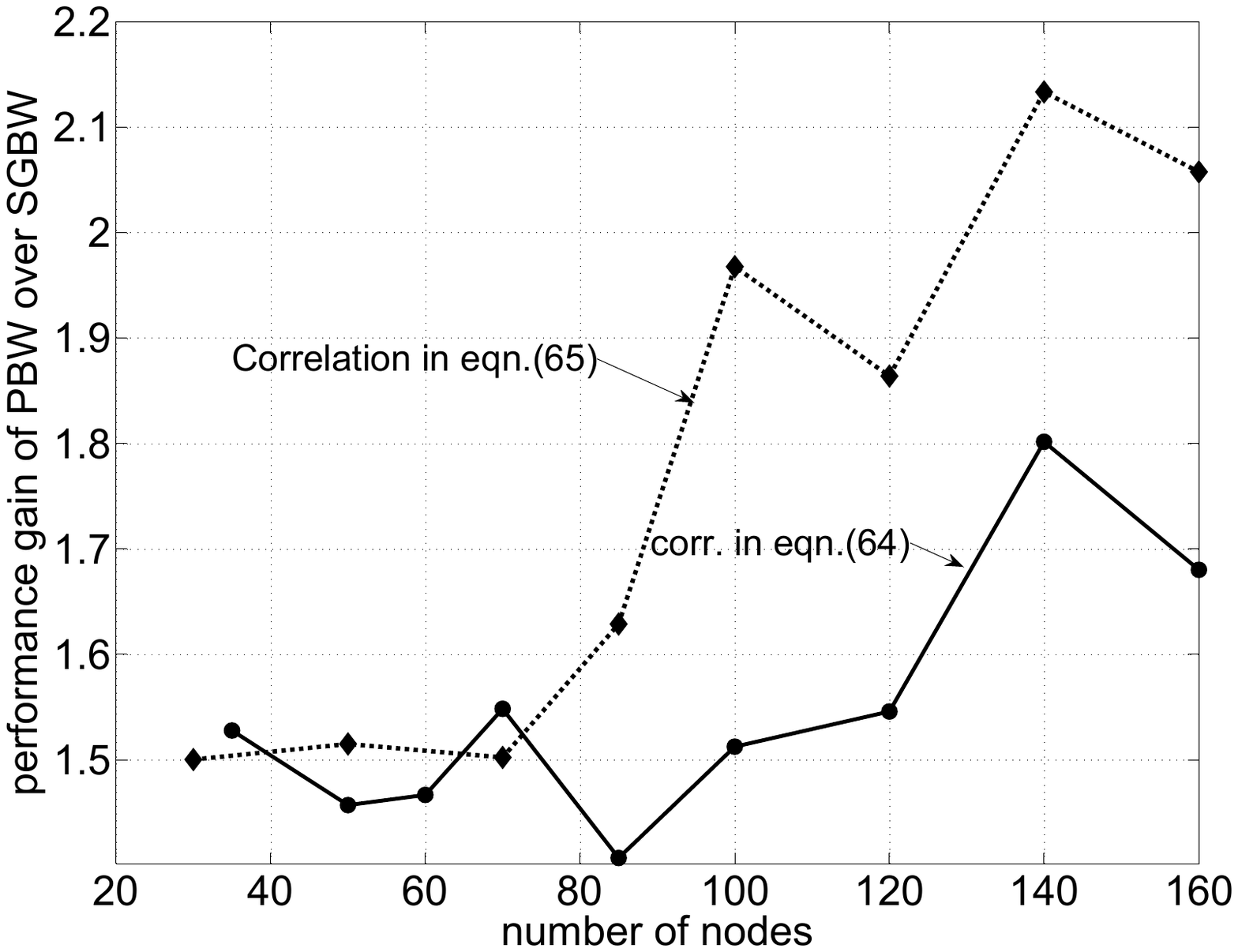}
      \caption{Performance gain of PBW over SGBW ($\Gamma_s^{\eta}$, eqn.~\eqref{eqn_gamma_s_m_eta})
      as a function of the number of nodes in the network.}
      \label{figure_SGBW_PBW}
\end{figure}
\subsection{Broadcast gossip algorithm~\cite{scaglione}: Asymmetric random links
}
\label{subsect_broadcast_gossip}

In the previous section, we demonstrated the effectiveness of our approach in networks with random symmetric link failures. This section demonstrates the validity of our approach in randomized protocols with asymmetric links.
We study the broadcast gossip algorithm~\cite{scaglione}. Although
the optimization problem~\eqref{Eqn_asym_optimization_problem} is convex for generic spatially correlated directed random links,
 we pursue here numerical optimization of the broadcast gossip algorithm proposed in~\cite{scaglione}, where, at each time step, node $i$
 is selected at random, with probability $1/N$. Node $i$ then broadcasts its state to all its neighbors within its
 wireless range. The neighbors then update their state by performing
  the weighted average of the received state with their own state. The nodes outside the set $\Omega_i$ and the node $i$
   itself keep their previous state unchanged. The broadcast gossip algorithm is well suited for WSN applications, since
it exploits the broadcast nature of wireless media and avoids bidirectional communication~\cite{scaglione}.

Reference~\cite{scaglione} shows that, in broadcast gossiping, all the nodes converge a.s. to a common random value $c$ with  mean $x_{\mbox{\scriptsize{avg}}}$ and bounded mean squared error.
Reference~\cite{scaglione} studies the case when the weights $W_{ij}=g$, $\forall (i,j) \in E$ and finds the optimal $g=g^{*}$ that
 optimizes the mean square deviation MSdev (see eqn.~\eqref{Eqn_asym_optimization_problem}).
We optimize the same objective function (see eqn.~\eqref{Eqn_asym_optimization_problem}) as in~\cite{scaglione}, but allowing different weights for different directed links. We detail on the numerical optimization for the broadcast gossip in the Appendix C.
We consider again the supergraph $G$ from our standard experiment with $N=100$ and average degree $15\%N$. For the broadcast gossip, we compare the performance of PBW with 1) the optimal equal weights in~\cite{scaglione} with $W_{ij}=g^{*}$, $(i,j) \in E$; 2) broadcast gossip with $W_{ij}=0.5$, $(i,j) \in E$.

Figure~3 (left) plots the consensus mean square deviation $\mathrm{MSdev}$ for the 3 different weight choices. The decay of MSdev is much faster for the PBW than for $W_{ij}=0.5$, $\forall \,(i,j)$ and $W_{ij}=g^{*}$, $\forall \,(i,j)$. For example, the MSdev  falls below $10\%$ after 260 iterations for PBW (i.e., 260 broadcast transmissions); broadcast gossip with $W_{ij}=g^{*}$ and $W_{ij}=0.5$ take 420 transmissions to achieve the same precision. This is to be expected, since PBW has many moredegrees of freedom for to optimize than the broadcast gossip in~\cite{scaglione} with all equal weights $W_{ij}=g^{*}$.
Figure 3 (right) plots the MSE, i.e., the deviation of the true average $x_{\mbox{\scriptsize{avg}}}$, for the three weight choices. PBW shows faster decay of
MSE than the broadcast gossip with $W_{ij}=g^{*}$ and $W_{ij}=0.5$.
\begin{figure}[thpb]
      \centering
      \includegraphics[height=2.2
      in,width=3.0in]{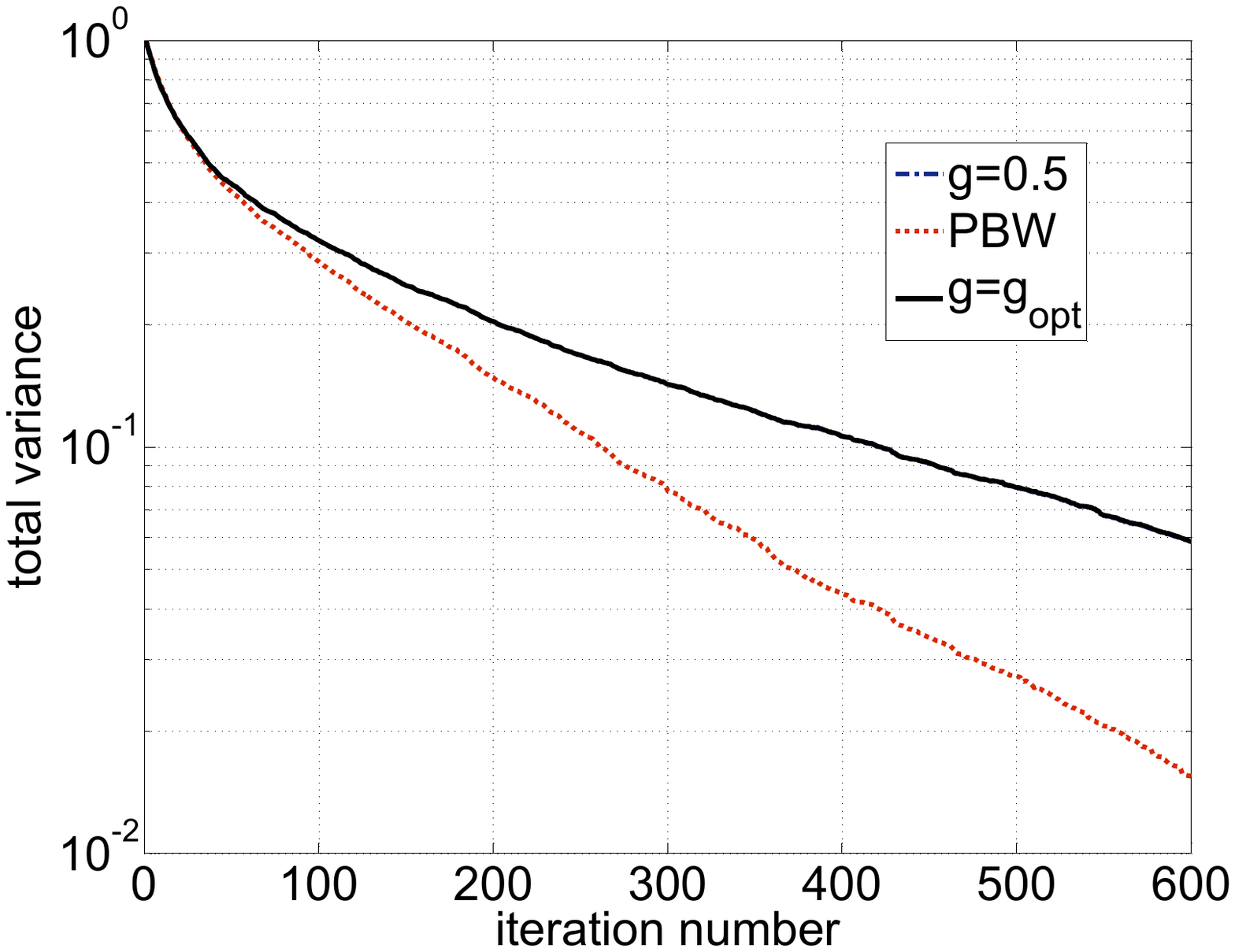}
      \includegraphics[height=2.14in,width=2.93in]{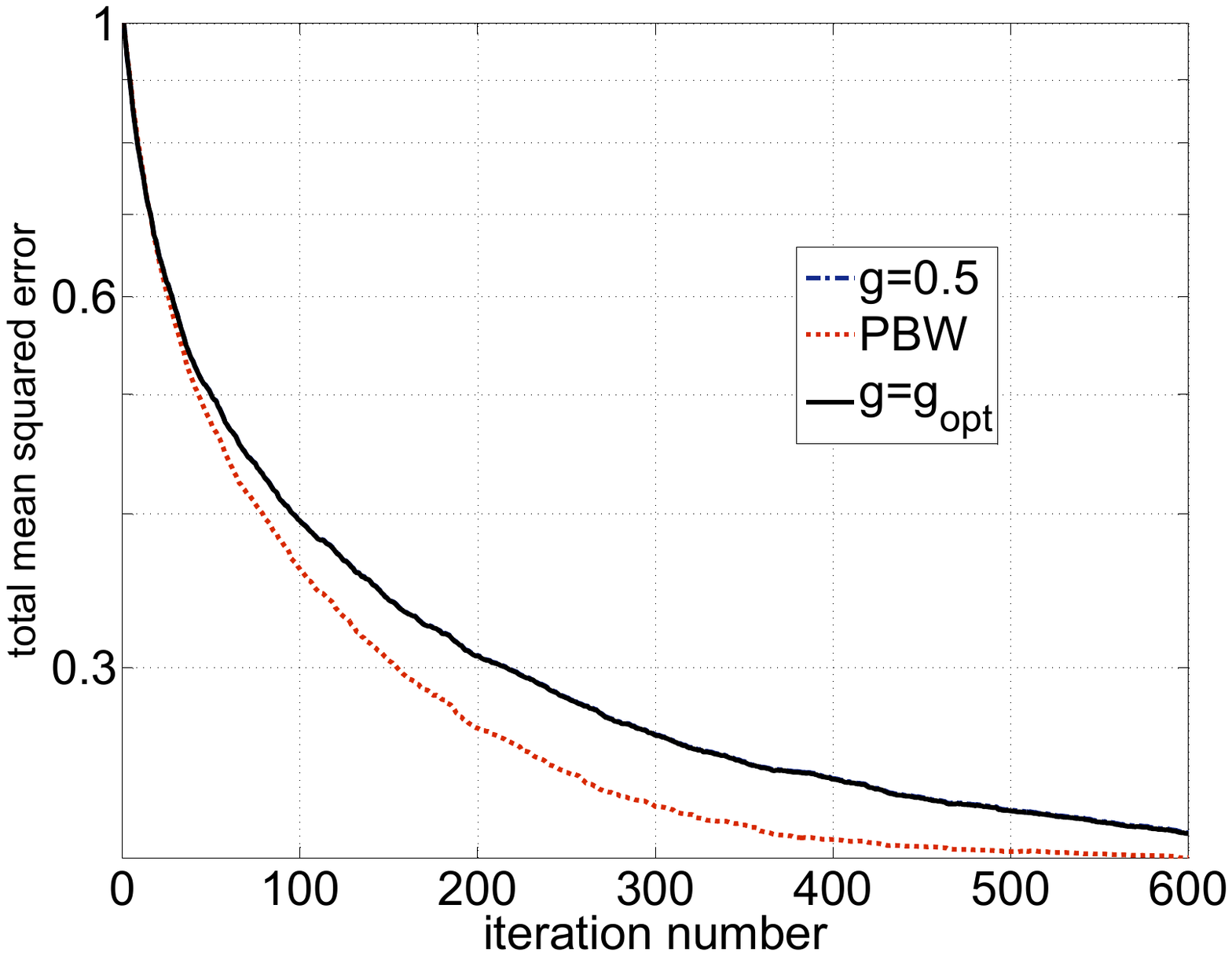}
      \label{FigureBroadcastGossipVar}
      \caption{Broadcast gossip algorithm with different weight choices. Left: total variance; right: total mean squared error}
      \label{FigureBroadCastGossipMSE}
\end{figure}
The weights provided by PBW are different among themselves, varying from 0.3 to 0.95. The weights $W_{ij}$ and $W_{ji}$ are
 also different, where the maximal difference between $W_{ij}$ and $W_{ji}$, $(i,j) \in E$, is 0.6. Thus, in the case of directed random networks, asymmetric matrix $W$ results in faster
 convergence rate.

\section{Conclusion}
\label{ConclusionSection}
In this paper, we studied the optimization of the weights for the consensus algorithm under random topology and spatially correlated links. We considered both networks with random link failures and randomized algorithms; from the weights optimization point of view, both fit into the same framework. We showed that, for symmetric random links, optimizing the MSE convergence rate is a convex optimization problem, and , for asymmetric links, optimizing the mean squared deviation from the current average state is also a convex optimization problem. We illustrated with simulations that the probability based weights (PBW) outperform previously proposed weights strategies that do not use the statistics of the network randomness. The simulations also show that, using the link quality estimates and the link correlations for designing the weights significantly improves the convergence speed, typically reducing the time to consensus by one third to a half, compared to choices previously proposed in the literature.
\appendices

\section{Proof of Lemma~\ref{LemmaStatisticsOfW} (a sketch)}
\label{AppendixSection}
Eqn.~(\ref{eqn:meancalW}) follows from the expectation of~(\ref{WEIGHTMATRIXCOMPACT}). To prove the remaining of the Lemma, we find $\mathcal{W}^2$, $\overline{\rule{0pt}{9pt}W}^2$, and the expectation $\mathcal{W}^2$. We obtain successively:
\begin{eqnarray*}
\mathcal{W}^2 & = & (\,\,W \odot \mathcal{A} + I - \mathrm{diag}(W\,\mathcal{A})\,\,)^2 \\
&=&       (\,W \odot \mathcal{A}\,)  ^2  +     \,\mathrm{diag}^2(W\,\mathcal{A})\,
+  I                   +  2 \,W \odot \mathcal{A}    -2 \,\mathrm{diag}(W\,\mathcal{A}) -  \, (\,W \odot \mathcal{A}\,)\,\, \mathrm{diag}(W\,\mathcal{A})\\
&-&
\, \,\mathrm{diag}(W\,\mathcal{A})  \, (\,W \odot \mathcal{A}\,)\\
\overline{\rule{0pt}{9pt}W}^2    &  =  &   (\,W \odot P\,)  ^2  +
\,\mathrm{diag}^2(W\,P)\,
+  I
         +2 \,W \odot P    -2
\,\mathrm{diag}\left( W\,P \right) \\
& \,&-\left[(\,W \odot P\,)\,\, \mathrm{diag}\left(W\,P\right)\,+
\,\mathrm{diag}\left(W\,P\right)  \, (\,W \odot P\,)\right]
\\
 \mathrm{E}\left[\,\mathcal{W}^2\,\right] & = &  \mathrm{E}\left[ (\,W \odot \mathcal{A}\,)  ^2 \right] +    \mathrm{E}\left[
\,\mathrm{diag}^2(W\,\mathcal{A})\right]
+  I         +2 \,W \odot P    \\
& \, & -2 \,\mathrm{diag}\left( W\,P \right) -
 \mathrm{E}[\, (\,W \odot \mathcal{A}\,)\, \mathrm{diag}(W\,\mathcal{A})\,+
\,\mathrm{diag}(W\,\mathcal{A})  \, (\,W \odot \mathcal{A}\,) \,]
\end{eqnarray*}
We will next show the following three equalities:
\begin{eqnarray}
\label{eqn1_proof_lemma_1}
 \mathrm{E}\left[(W \odot \mathcal{A})^2\right]\hspace{-.3cm}&=&\hspace{-.3cm}
 (W \odot P)  ^2 + {W_C}^T\left\{{R_A}\odot(11^T\otimes I) \right\}{W_C}  \\
\label{eqn2_proof_lemma_1}
\mathrm{E}\left[\mathrm{diag}^2\left(W\mathcal{A}\right)\right] \hspace{-.3cm}&=&\hspace{-.3cm}
 \mathrm{diag}^2(W\,P)  +  {W_C}^T\left\{{R_A}\odot (I\otimes 11^T) \right\}{W_C}  \\
 \label{eqn3_proof_lemma_1}
\mathrm{E}\left[(W \odot \mathcal{A})
\mathrm{diag}\left(W\mathcal{A}\right)+
\mathrm{diag}\left(W\mathcal{A} \right)(W \odot \mathcal{A}) \right]
\hspace{-.3cm}&=&\hspace{-.3cm}\\
(W \odot P)\mathrm{diag}\left( W P \right)\hspace{-.3cm}&+&\hspace{-.3cm}
\mathrm{diag}\left( W P \right)(W \odot P) -
{W_C}^T \left\{{R_A}\odot B \right\}{W_C}  \nonumber
\end{eqnarray}
%
First, consider~(\ref{eqn1_proof_lemma_1}) and find $\mathrm{E}\left[\left(\,W \odot \mathcal{A}\,\right)^2\right]$. Algebraic manipulations allow to write $(\,W \odot \mathcal{A}\,)^2$ as follows:
\begin{equation}
\label{eqn_proof_lemma1}
\left(\,W \odot \mathcal{A}\,\right)^2 = {W_C}^T \left\{  \mathcal{A}_2  \odot (\, 11^T \otimes I \,)  \right\}  {W_C}, \:\:
\mathcal{A}_2  = \mathrm{Vec}(\, \mathcal{A}\,)  \mathrm{Vec}
^T(\, \mathcal{A}\,)
\end{equation}
%
%
%
To compute the expectation of~(\ref{eqn_proof_lemma1}),
 we need $\mathrm{E}  \left[\,\mathcal{A}_2  \,\right]$ that can be written as
\[
\mathrm{E}  \left[\,\mathcal{A}_2  \,\right]=P_2 + R_A,\:\:\mbox{with}\:\:
P_2  =  \,\, \mathrm {Vec}(\,P\, ) \,\,
 \mathrm{Vec}^T(\,P\,).
 \]
%
%
Equation~(\ref{eqn1_proof_lemma_1}) follows, realizing that 
\[
{W_C}^T \left\{  P_2 \odot (\,
11^T \otimes I \,)\right\} \,{W_C} = (\,W \odot P\,)  ^2.
\]
Now consider~(\ref{eqn2_proof_lemma_1}) and~(\ref{eqn3_proof_lemma_1}). After algebraic manipulations, it can be shown that
 \begin{eqnarray*}
 \mathrm{diag}^2\left( W\,\mathcal{A} \right)   &=&  {W_C}^T \left\{ \mathcal{A}_2  \odot (\, I \otimes 11^T \,)\right\}{W_C} \\
(\,W \odot \mathcal{A}\,)\,\, \mathrm{diag}\left( W\,\mathcal{A} \right)\,+ \,\mathrm{diag}\left(W\,\mathcal{A}\right)  \, (\,W \odot \mathcal{A}\,) &=&   {W_C}^T \left\{ \mathcal{A}_2  \odot B\right\}{W_C}
\end{eqnarray*}
Computing the expectations in the last two equations leads to eqn.~(\ref{eqn2_proof_lemma_1})  and eqn.~(\ref{eqn3_proof_lemma_1}).

Using equalities~(\ref{eqn1_proof_lemma_1}), (\ref{eqn2_proof_lemma_1}), and~(\ref{eqn3_proof_lemma_1}) and comparing the expressions for $\overline{\rule{0pt}{9pt}W}^2$ and $\mathrm{E}[\,\mathcal{W}^2\,]$ leads to:
  \begin{equation}
  \label{eqnapp:RC}
  R_C=\mathrm{E}[\,\mathcal{W}^2\,] -  \overline{\rule{0pt}{9pt}W}^2=
   {W_C}^T\,\,\{\,\,{R_A}\odot (  I \otimes 11^T  +   11^T \otimes I  -  B) \}\,\,{W_C}
  \end{equation}
This completes the proof of Lemma~\ref{LemmaStatisticsOfW}.

\section{Subgradient step calculation for the case of spatially correlated links}
To compute the subgradient~$H$, from eqns.~(\ref{eq:PBW}) and~(\ref{eq:PBW-b}) we consider the computation of $\mathrm{E}\left[\mathcal{W}^2 - J\right]=\overline{\rule{0pt}{9pt}W}^2-J+R_C$. Matrix $\overline{\rule{0pt}{9pt}W}^2-J$ is computed in the same way as for the uncorrelated case. To compute $R_C$, from~(\ref{eqnapp:RC}),
partition the matrix $R_A$ into $N\times N$ blocks:
{\small
\begin{equation*}
R_A =  \left( \begin{array}{cccc}
R_{11} & R_{12} & \ldots & R_{1N}  \\
R_{21} & R_{22} & \ldots & R_{2N}  \\
\vdots & \ldots & \ldots & \vdots  \\
R_{N1} & R_{N2} & \ldots & R_{NN}
\end{array} \right)
\end{equation*}
}
Denote by $d_{ij}$, by $c_{ij}^l$, and by $r_{ij}^l$ the diagonal,  the $l$-th column, and the $l$-th row  of the block $R_{ij}$.
It can be shown that the matrix $R_C$ can be computed as follows:
\begin{eqnarray*}
 \left[ R_C \right]_{ij} & = & W_i^T \left(   d_{ij} \odot W_j   \right)  -  W_{ij} \left(     W_i^T c_{ij}^i + W_j^T r_{ij}^j   \right),\,\,i \neq j\\
  \left[ R_C \right]_{ii} & = & W_i^T \left(   d_{ii} \odot W_i   \right)  +  W_i^T R_{ii} W_i
\end{eqnarray*}

Denote by $R_A(:, k)$ the $k$-th column of the matrix $R_A$ and by
\begin{eqnarray*}
k_1 & = &  \left(  e_j^T \otimes I_N \right) R_A(:, (i-1)N+j),\,\,k_2  =   \left(  e_i^T \otimes I_N \right) R_A(:, (j-1)N+i) , \\
k_3 & = & \left(  e_i^T \otimes I_N \right) R_A(:, (i-1)N+j),\,\,
k_4  =  \left(  e_j^T \otimes I_N \right) R_A(:, (j-1)N+i).
\end{eqnarray*}
Quantities $k_1$, $k_2$, $k_3$ and $k_4$ depend on $(i,j)$ but for the sake of the notation simplicity indexes are omitted.
It can be shown that the computation of $H_{ij}$, $(i,j) \in E$ boils down to:
\begin{equation*}
H_{ij} = 2 \,u_i^2 \,W_i^T \, c_{ii}^j  +  2 \, u_j^2 \, W_j^T \, c_{jj}^i + 2 \, u_i  \, W_j^T \,  (u \odot k_1) \, + \, 2 \, u_j \,  W_i^T\,   (u \odot k_2) \, -2\,  u_i \, u_j \, W_j^T \, c_{ji}^j - $$ $$2 \, u_i\, u_j \, W_i^T \, c_{ij}^i     -2 \, u_i \, W_i^T \,  (u \odot k_3)\,  - 2 \, u_j \, W_j^T  \, (u \odot k_4) +2 \, {P}_{ij} \, (u_i-u_j) \, u^T \left(  \overline{\rule{0pt}{9pt}W}_j \, - \, \overline{\rule{0pt}{9pt}W}_i  \, \right)
\end{equation*}
\section{Numerical optimization for the broadcast gossip algorithm}
With broadcast gossip, the matrix $\mathcal{W}(k)$ can take $N$ different realizations, corresponding to the broadcast cycles of each of the $N$ sensors. We denote these realizations by $\mathcal{W}^{(i)}$, where $i$ indexes the broadcasting node. We can write the random realization of the broadcast gossip matrix $\mathcal{W}^{(i)}$, $i=1,...,N$, as follows:
\begin{equation}
\mathcal{W}^{(i)}(k) = W \odot \mathcal{A}^{(i)}(k) +I - \mathrm{diag} \left(  W \, \mathcal{A}^{(i)}(k) \right),
\end{equation}
where $\mathcal{A}^{(i)}_{li}(k) = 1$, if $l \in \Omega_i$. Other entries of $\mathcal{A}^{(i)}(k)$ are zero.

Similarly in Appendix A, we can arrive at the expressions for $\mathrm{E} \left[ \mathcal{W}^T \mathcal{W} \right]:=\mathrm{E} \left[ \mathcal{W}^T(k) \mathcal{W}(k)  \right]$ and for
$\mathrm{E} \left[ \mathcal{W}^T J \mathcal{W}  \right]:=\mathrm{E} \left[ \mathcal{W}^T(k) J \mathcal{W}(k)  \right]$, for all $k$. We remark that the matrix $W$ needs not to be symmetric for the broadcast gossip and that $W_{ij}=0$, if $(i,j) \notin E$.
\begin{eqnarray*}
\mathrm{E}\, \left[ \left( \mathcal{W} ^T \mathcal{W }\right)_{ii} \right]&=& \frac{1}{N} \sum_{l=1, l\neq i}^N W_{li}^2
+   \frac{1}{N}  \sum_{l=1, l \neq i}^N (1 - W_{il})^2\\
\mathrm{E}\, \left[ \left( \mathcal{W} ^T \mathcal{W }\right)_{ij} \right]&=& \frac{1}{N} W_{ij} (1-W_{ij}) + \frac{1}{N} W_{ji} (1 - W_{ji}),\,\, i\neq j \\
\mathrm{E} \left[ \left( \mathcal{W}^T J \mathcal{W} \right)_{ii}\right] & = &  \frac{1}{N^2} \left(    1 + \sum_{l \neq i}W_{li}    \right)^2 + \frac{1}{N^2}
\sum_{l=1,l\neq i}^N (1-W_{il})^2  \\
\left[ \mathrm{E}  \left[ \mathcal{W}^T J \mathcal{W}  \right)_{ij} \right] &=& \frac{1}{N^2} (1-W_{ji})
(1 + \sum_{l=1,l\neq i}^N W_{li}) + \frac{1}{N^2}  (1-W_{ij}) (1 + \sum_{l=1,l \neq j }^N W_{lj}) \\
&+&
\frac{1}{N^2} \sum_{l=1,l \neq i, l\neq j}^N (1-W_{il}) (1-W_{jl}), \,\,i \neq j
\end{eqnarray*}

Denote by $W^{\mathrm {BG}}:=\mathrm{E} \left[ \mathcal{W}^T \mathcal{W}  \right] - \mathrm{E} \left[ \mathcal{W}^T J \mathcal{W}  \right]$ and recall the definition of the MSdev rate $\psi(W)$~\eqref{PSI_FUNCTION}. We have that $\psi(W)=\lambda_{\mathrm{max}}\left(
W^{\mathrm {BG}}    \right)$. We proceed with the calculation of the subgradient of $\psi(W)$ similarly as in subsection IV-E.
The partial derivative of the cost function $\psi(W)$ with respect to weight $W_{i,j}$ is given by:
$$
\frac {\partial} {\partial W_{i,j}} \lambda_\mathrm {max} \left(W^{\mathrm {BG}} \right)= q^T \left( \frac {\partial} {\partial W_{i,j}} W^{\mathrm {BG}} \right)q$$ where $q$ is eigenvector associated with the maximal eigenvalue of the matrix $W^{\mathrm {BG}}$.
Finally, partial derivatives of the entries of the matrix $W^{\mathrm {BG}}$ with respect to weight $W_{i,j}$ are given by the following set of equations:
\begin{eqnarray*}
\frac {\partial} {\partial W_{i,j}} W^{\mathrm {BG}}_{i,i}&=& - 2 \frac{N-1}{N}(1-W_{i,j})\\
\frac {\partial} {\partial W_{i,j}} W^{\mathrm {BG}}_{j,j}&=& \frac{2}{N}W_{i,j} -\frac{2}{N} (1-\sum_{l=1,l\neq j}^N W_{l,j})\\
\frac {\partial} {\partial W_{i,j}} W^{\mathrm {BG}}_{i,j}&=& \frac{1}{N}(1-2 W_{i,j})  -   \frac{1}{N^2} (-1-\sum_{l=1,l\neq j}^N W_{l,j}-W_{i,j})\\
\frac {\partial} {\partial W_{i,j}} W^{\mathrm {BG}}_{i,l}&=& \frac{1}{N^2}(1-W_{l,j}), \, l\neq i, l\neq j \\
\frac {\partial} {\partial W_{i,j}} W^{\mathrm {BG}}_{i,j}&=& -\frac{1}{N^2}(1-W_{l,j}), \, l\neq i, l\neq j\\
\frac {\partial} {\partial W_{l,m}} W^{\mathrm {BG}}_{i,j}&=& 0, \,\, otherwise.
\end{eqnarray*}
%
%
%
%
%
%
%
%
%
%

\bibliographystyle{IEEEtran}
\bibliography{IEEEabrv,Bibliography_Dusan_30_April}

\end{document}